\documentclass[twoside,leqno]{article}
\usepackage{geometry}
\usepackage[english]{babel} 
\usepackage{times} 
\usepackage{graphicx} 
\usepackage{amsmath}
\usepackage{amsthm}
\usepackage{amssymb}
\usepackage{wrapfig}
\usepackage{xcolor}

\usepackage{microtype}
\usepackage{color}

\graphicspath{{img/}}

\definecolor{darkblue}{rgb}{0,0,0.4}
\definecolor{darkred}{rgb}{0.5,0,0}

\usepackage[pdftex,colorlinks=true,linkcolor=darkblue,citecolor=darkred,urlcolor=blue]{hyperref}

\usepackage[pagewise]{lineno}

\usepackage{xspace}
\usepackage{gensymb}
\usepackage{paralist}
\usepackage[textsize=small]{todonotes}

\newcommand{\rephrase}[3]{\noindent\textbf{#1 #2}.~\emph{#3}}

\newcommand{\remove}[1]{}

\newtheorem{corollary}{Corollary}

\definecolor{blue}{rgb}{0.274,0.392,0.666}
\definecolor{red}{rgb}{0.627,0.117,0.156}
\definecolor{green}{rgb}{0,0.588,0.509}

\newcommand{\red}[1]{{\color{red}{#1\xspace}}}
\newcommand{\blue}[1]{{\color{blue}{#1\xspace}}}

\newtheorem{theorem}{Theorem}
\newtheorem{lemma}{Lemma}
\newtheorem{claim}{Claim}

\usepackage{subfigure}

\begin{document}

\newcommand{\cpp}{{\sc C-Planarity with Pipes}\xspace}
\newcommand{\cp}{{\sc C-Planarity}\xspace}
\newcommand{\calA}[1]{\ensuremath{{\cal A}^{#1}}}
\newcommand{\calB}[1]{\ensuremath{{\cal B}^{#1}}}
\newcommand{\calC}[1]{\ensuremath{{\cal C}^{#1}}}
\newcommand{\calD}[1]{\ensuremath{{\cal D}^{#1}}}
\newcommand{\calE}[1]{\ensuremath{{\cal E}^{#1}}}
\newcommand{\calF}[1]{\ensuremath{{\cal F}^{#1}}}
\newcommand{\calG}[1]{\ensuremath{{\cal G}^{#1}}}
\newcommand{\calH}[1]{\ensuremath{{\cal H}^{#1}}}
\newcommand{\calI}[1]{\ensuremath{{\cal I}^{#1}}}
\newcommand{\calL}[1]{\ensuremath{{\cal L}^{#1}}}
\newcommand{\calT}[1]{\ensuremath{{\cal T}^{#1}}}
\newcommand{\calO}[1]{\ensuremath{{\cal O}^{#1}}}
\newcommand{\calP}[1]{\ensuremath{{\cal P}^{#1}}}
\newcommand{\calQ}[1]{\ensuremath{{\cal Q}^{#1}}}
\newcommand{\calX}[1]{\ensuremath{{\cal X}^{#1}}}
\newcommand{\calY}[1]{\ensuremath{{\cal Y}^{#1}}}
\newcommand{\calZ}[1]{\ensuremath{{\cal Z}^{#1}}}

\newcommand{\Gr}[1]{\ensuremath{\red{G^{#1}_1}}\xspace}
\newcommand{\Gb}[1]{\ensuremath{\blue{G^{#1}_2}}\xspace}
\newcommand{\EmbR}[1]{\ensuremath{\red{\mathcal{E}^{#1}_1}}\xspace}
\newcommand{\EmbB}[1]{\ensuremath{\blue{\mathcal{E}^{#1}_2}}\xspace}
\newcommand{\GammaR}[1]{\ensuremath{\red{\Gamma^{#1}_1}}\xspace}
\newcommand{\GammaB}[1]{\ensuremath{\blue{\Gamma^{#1}_2}}\xspace}
\newcommand{\GExtr}[1]{\ensuremath{\red{G_1(V^{#1},E^{#1}_1)}}\xspace}
\newcommand{\GExtb}[1]{\ensuremath{\blue{G_2(V^{#1},E^{#1}_2)}}\xspace}
\newcommand{\Er}[1]{\ensuremath{\red{E^{#1}_1}}\xspace}
\newcommand{\Eb}[1]{\ensuremath{\blue{E^{#1}_2}}\xspace}
\newcommand{\Gint}[1]{\ensuremath{G^{#1}_\cap}\xspace}
\newcommand{\sefeinstance}[1]{$\langle \Gr{#1},\Gb{#1}\rangle$\xspace}
\newcommand{\sefeinstanceextended}[1]{$\langle \red{G^{#1}_1}=(V^{#1},\red{E^{#1}_1}),\blue{G^{#1}_2}=(V^{#1},\blue{E^{#1}_2})\rangle$\xspace}
\newcommand{\sefesolution}[1]{$\langle \red{\Gamma_1^{#1}},\blue{\Gamma_2^{#1}}\rangle$\xspace}

\newcommand{\skel}{sk}
\newcommand{\pert}{pert}
\newcommand{\pertinent}{pert}

\newcommand{\cgraph}[1]{\ensuremath{\calC{#1}(G^{#1},\calT{#1})}\xspace}

\newcommand{\iccp}{{\sc Inclusion-Constrained C-Planarity}\xspace}
\newcommand{\ctree}{{components tree}\xspace}
\newcommand{\nctree}{{neighbor-clusters tree}\xspace}
\newcommand{\ctrees}{{components trees}\xspace}
\newcommand{\nctrees}{{neighbor-clusters trees}\xspace}
\newcommand{\Nctrees}{{Neighbor-clusters trees}\xspace}
\newcommand{\Bnctree}{{Block-neighbor-clusters tree}\xspace}
\newcommand{\Nctree}{{Neighbor-clusters tree}\xspace}
\newcommand{\bnctree}{{block-neighbor-clusters tree}\xspace}
\newcommand{\bnctrees}{{block-neighbor-clusters trees}\xspace}
\newcommand{\enctree}{{pipe-neighbor-clusters tree}\xspace}
\newcommand{\enctrees}{{pipe-neighbor-clusters trees}\xspace}

\newcommand{\me}{{multi-edge}\xspace}
\newcommand{\se}{{single-edge}\xspace}

\newcommand{\ga}{{clusters-adjacency graph}\xspace}

\title{Clustered Planarity with Pipes%
\thanks{This article reports on work supported by the U.S.~Defense Advanced
	Research Projects Agency (DARPA) under agreement no.~AFRL
	FA8750-15-2-0092.
	The views expressed are those of the authors and do not reflect the
	official policy or position of the Department of Defense
	or the U.S.~Government.
	This work was partially supported by DFG grant Ka812/17-1 and by MIUR Project ``AMANDA'' 2012C4E3KT.}
}

\author{Patrizio Angelini
\thanks{Wilhelm-Schickard-Institut f\"ur Informatik, Universit\"at T\"ubingen, Germany, \href{mailto:angelini@informatik.uni-tuebingen.de}{\em angelini@informatik.uni-tuebingen.de}.
}
\and 
Giordano {Da Lozzo}
\thanks{Department of Engineering, Roma Tre University, Italy. 
\href{mailto:dalozzo@dia.uniroma3.it}{\em dalozzo@dia.uniroma3.it}.
}
}

\date{}
\maketitle

\begin{abstract} 
We study the version of the {\sc C-Planarity} problem in which edges connecting the same pair of clusters must be grouped into pipes, which generalizes the {\sc Strip Planarity} problem. We give algorithms to decide several families of instances for the two variants in which the order of the pipes around each cluster is given as part of the input or can be chosen by the algorithm.
\end{abstract}

\section{Introduction} \label{se:introduction}

Visualizing clustered graphs is a challenging task with several applications in the analysis of networks that exhibit a hierarchical structure.  
The most established criterion for a readable visualization of these graphs has been formalized in the notion of {\em c-planarity}, introduced by Feng, Cohen, and Eades~\cite{FengCE95} in 1995. Given a {\em clustered graph} $\cgraph{}$ ({\em c-graph}), that is, a graph $G$ equipped with a recursive clustering $\calT{}$ of its vertices, the \cp problem asks whether there exist a planar drawing of $G$ and a representation of each cluster as a topological disk enclosing all and only its vertices, such that no {\em ``unnecessary''} crossings occur between disks and edges, or between disks.
Ever since its introduction, this problem has been attracting a great deal of research. However, the question regarding its computational complexity withstood the attack of several powerful algorithmic tools, such as the Hanani-Tutte theorem~\cite{f-thttcg-14,s-ttphtpv-13}, the SPQR-tree machinery~\cite{cdfpp-ccccg-08}, and the Simultaneous PQ-ordering framework~\cite{br-npcpcep-14}. 

The clustering of a c-graph \cgraph{} is described by a rooted tree $\calT{}$ whose leaves are the vertices of $G$ and whose each internal node $\mu$ different from the root represents a \emph{cluster} containing all and only the leaves of the subtree of $\calT{}$ rooted at $\mu$. A c-graph is {\em flat} if $\calT{}$ has height $2$. The {\em \ga} $G_A$ of a flat c-graph is the graph obtained from the c-graph by contracting each cluster into a single vertex, and by removing multi-edges and loops.

Cortese {\em et al.}~\cite{cdpp-ecpg-09} introduced a variant of \cp for flat c-graphs, which we call {\sc C-Planarity with Embedded Pipes}. The input of this problem is a flat c-graph \cgraph{} together with a planar drawing of its \ga $G_A$, in which vertices of $G_A$ are represented by disks and edges of $G_A$ by pipes connecting the disks. The goal is then to produce a c-planar drawing of \cgraph{} in which each vertex of $G$ lies inside the disk representing the cluster it belongs to and each inter-cluster edge of $G$ is drawn inside the corresponding pipe. In~\cite{cdpp-ecpg-09} this problem is solved when the underlying graph $G$ is a cycle. Chang, Erickson, and Xu~\cite{weakly-15} observed that in this case the problem is equivalent to determining whether a closed walk of length $n$ in a simple plane graph is weakly simple, and improved the time complexity to $O(n\log{n})$.
The special case of the problem in which the \ga is a path while $G$ can be any planar graph, which is known by the name of {\sc Strip Planarity}, has also been studied. Polynomial-time algorithms for this problem have been presented when the underlying graph has a fixed planar embedding~\cite{strip-gd-13} and when it is a tree~\cite{f-thttcg-14}.
 
We remark that polynomial-time algorithms for the \cp problem are known when strong limitations on the number or on the arrangement of the components of the clusters are imposed, where a {\em component} of a cluster $\mu \in \cal T$ is a maximal connected subgraph induced by the vertices of $\mu$. In particular, \cp can be decided in linear time when each cluster contains one component~\cite{cdfpp-ccccg-08,FengCE95} (the c-graph is {\em c-connected}). However, even when each cluster contains at most two components, polynomial-time algorithms are known only when further restrictions are imposed on the c-graph~\cite{br-npcpcep-14,jjkl-cpecgtcc-08}.
The results we show in this paper are also based on imposing constraints on the number and combination of certain types of components.

A component of a cluster $\mu \in \mathcal{T}$ is {\em \me} if it is incident to at least two inter-cluster edges, otherwise it is {\em \se}.
Also, it is {\em passing} if it is adjacent to vertices belonging to at least two clusters in $\cal T$ different from $\mu$. Otherwise, it is adjacent to vertices of a unique cluster $\nu \in \mathcal{T}$ different from $\mu$; in this case, we say that it is {\em originating from $\mu$ to $\nu$}. For {\sc Strip Planarity} the originating components can be further distinguished into {\em source} and {\em sink} components, based on whether $\nu$ corresponds to the strip above of below the one of $\mu$.

{\bf Our contributions.}  We show that {\sc Strip Planarity} is polynomial-time solvable for instances with a unique source component (Section~\ref{se:strips}) and that {\sc C-Planarity with Embedded Pipes} is polynomial-time solvable for instances such that, for each cluster $\mu \in \cal T$ and for each edge $(\mu,\nu)$ in $G_A$, either 
cluster $\mu$ contains at most one originating \me component from $\mu$ to $\nu$,
or it contains at most two \me originating components from $\mu$ to $\nu$ and does not contain any passing component that is incident to $\nu$
(Section~\ref{cpp:fixed-embedding}). Finally, in Section~\ref{se:cpp-fpt} we introduce a generalization of {\sc C-Planarity with Embedded Pipes}, which we call \cpp. Given a c-graph \cgraph{}, the goal of this problem is to find a planar drawing of the \ga of \cgraph{} whose vertices and edges are represented by disks and pipes, respectively, that allows for a drawing of \cgraph{} that is a solution of {\sc C-Planarity with Embedded Pipes}. In other words, the goal is to find a c-planar drawing of \cgraph{} in which the inter-cluster edges are still required to be grouped into pipes, but the order of the pipes around each disk is not prescribed by the input. By introducing a new characterization of \cp, we give an FPT algorithm for \cpp that runs in $g(K,c) \cdot O(n^2)$ time, with $g(K,c) \in O(K^{c(K-2)})$, where $K$ is the maximum number of \me components in a cluster and $c$ is the number of clusters with at least two \me components. We remark that our results imply polynomial-time testing algorithms for all the three problems in the case in which each cluster contains at most two components.

\section{Preliminaries}\label{se:preliminaries}

For the standard definitions on planar graphs, planar drawings, planar embeddings, and connectivity we point the reader to~\cite{dett-gd-99}.
We use the term {\em rotation scheme} to denote the clockwise circular ordering of the edges incident to each vertex in a planar embedding, and refer to the containment relationships between vertices and cycles of the graph in the embedding as {\em relative positions}. Further, we say that a block of a $1$-connected graph is {\em trivial} if it consists of a single edge, otherwise it is {\em non-trivial}.

%

\smallskip
{\bf PQ-trees.} 
A \emph{PQ-tree} $T$ is an unrooted tree whose leaves are the elements of a ground set $A$. 
The internal nodes of $T$ are either {\em P-nodes} or {\em Q-nodes}. 
PQ-tree $T$ can be used to represent all and only the circular orderings ${\cal O}(T)$ on $A$ satisfying a given set of {\em consecutivity constraints} on $A$, each of which specifies that a subset of the elements of $A$ has to appear consecutively in all the sought circular orderings on $A$. 
The orderings in ${\cal O}(T)$ are all and only the circular orderings on the leaves of $T$ obtained by arbitrarily ordering the neighbours of each P-node and by arbitrarily selecting for each Q-node a given circular ordering on its neighbours or its reverse ordering. 
PQ-trees were originally introduced by Booth and Lueker~\cite{bt-tcopiggppqa-76} in a rooted version.

\smallskip
{\bf Connectivity.}
A {\em $k$-cut} of a graph is a set of at most $k$ vertices whose removal disconnects the graph. A connected graph is {\em biconnected} if it has no $1$-cut. The maximal biconnected components of a graph are its {\em blocks}.
Without loss of generality, in the following we assume that the \ga $G_A$ of \cgraph{} is connected and that, for every cluster $\mu \in \calT{}$ and for every component $c$ of $\mu$, it holds that:
\begin{inparaenum}[(i)]
	\item there exists at least an inter-cluster edge incident to $c$,
	\item every block of $c$ that is a leaf in the block-cut-vertex tree of $c$ contains at least a vertex $v$ such that $v$ is not a cut-vertex of $c$ and it is incident to at least an inter-cluster edge, and
	\item if there exists exactly one vertex in $c$ that is incident to inter-cluster edges, then $c$ consists of a single vertex.
\end{inparaenum}

\smallskip
{\bf Simultaneous Embedding with Fixed Edges.} 
Given planar graphs $\Gr{}=(V,\Er{})$ and $\Gb{}=(V,\Eb{})$, the {\sc SEFE} problem asks whether there exist planar drawings $\GammaR{}$ of $\Gr{}$ and $\GammaB{}$ of $\Gb{}$ such that (i) any vertex $v\in V$ is mapped to the same point in $\GammaR{}$ and $\GammaB{}$ and (ii) any edge $e \in \Er{} \cap \Eb{}$ is mapped to the same curve in $\GammaR{}$ and $\GammaB{}$. 
We call {\em common graph} and {\em union graph} the graphs $\Gint{} = (V, \Er{} \cap \Eb{})$ and $G_\cup = (V, \Er{} \cup \Eb{})$, respectively.
See~\cite{bkr-sepg-12} for a survey.
%


We state a theorem on SEFE that will be fundamental for our results. Even though this theorem has never been explicitly stated in the literature, it can be easily deduced from known results~\cite{br-spacep-13}, as discussed in the following.

\begin{theorem}\label{th:sefe}
	Let $\Gr{}=(V,\Er{})$ and $\Gb{}=(V,\Eb{})$ be two planar graphs whose common graph ${}\Gint{} = (V,\Er{} \cap \Eb{})$ is a forest and whose cut-vertices are incident to at most two non-trivial blocks. It can be tested in $O(|V|^2)$ time whether \sefeinstance{} admits a SEFE.
\end{theorem}

In particular, the theorem descends from a straightforward extension of the algorithm~\cite{br-spacep-13} to test SEFE of two biconnected planar graphs whose common graph is connected, to the case in which the common graph is a forest.

First, consider the following characterization of SEFE for two planar graphs.

\begin{theorem}[J{\"{u}}nger and Schulz~\footnote{M. J{\"{u}}nger and M. Schulz. Intersection graphs in simultaneous embedding with fixed edges. {\em J. Graph Algorithms Appl.}, 13(2):205–218, 2009.}, Theorem~4]\label{th:inutile}
	Two planar graphs $\GExtr{}$ and $\GExtb{}$ with common graph $\Gint{} = (V,\Er{} \cap \Eb{})$ have a SEFE if and only if they admit combinatorial embeddings inducing the same combinatorial embedding on $\Gint{}$.
\end{theorem}

Recall that a combinatorial embedding of a planar graph $G(V,E)$ is defined by (i) the rotation scheme of each vertex in $V$ and by (ii) the relative positions of the connected components of $G$. Hence, if $G$ is acyclic, then a combinatorial embedding of $G$ is entirely defined by (i). This fact and Theorem~\ref{th:inutile} imply the following.

\begin{corollary}\label{co:forest}
	Two planar graphs $\GExtr{}$ and $\GExtb{}$ whose common graph $\Gint{} = (V,\Er{} \cap \Eb{})$ is a forest have a SEFE if and only if they admit combinatorial embeddings inducing the same rotation scheme on $\Gint{}$.
\end{corollary}

Bl\"asius and Rutter proved~\cite{br-spacep-13} that it can be tested in quadratic time whether two biconnected planar graphs $\GExtr{}$ and $\GExtb{}$ admit combinatorial embeddings $\EmbR{}$ and $\EmbB{}$, respectively, such that the rotation scheme of each vertex in $V$ is the same in $\EmbR{}$ and in $\EmbB{}$, when restricted to the common edges. They also proved that such a result extends to the case in which $\Gr{}$ and $\Gb{}$ have cut-vertices incident to at most two non-trivial blocks.
Hence, Theorem~\ref{th:sefe} directly follows from Corollary~\ref{co:forest} and from the results in~\cite{br-spacep-13}.


\section{Single-Source Strip Planarity}\label{se:strips}

In this section we prove a result of the same flavour as that by Bertolazzi {\em et al.}~\cite{bdmt-ouptssd-98} for the upward planarity testing of single-source digraphs. Namely, we show that instances of {\sc Strip Planarity} 
with
a unique source component can be tested efficiently.
The {\sc Strip Planarity} problem takes in input a pair $\langle G=(V,E), \gamma \rangle$, where $G=(V,E)$ is a planar graph and $\gamma: V \rightarrow \{1,\dots,k\}$ is a mapping of each vertex to one of $k$ unbounded horizontal strips of the plane such that, for any two adjacent vertices $u,v \in V$, it holds that $|\gamma(u)  - \gamma(v)| \leq 1 $. The goal is to find a planar drawing of $G$ in which vertices lie inside the corresponding strips and edges cross the boundary of any strip at most once.
We observe that {\sc Strip Planarity} is equivalent to {\sc C-Planarity with Embedded Pipes} when $G_A$ is a path~\cite{strip-gd-13}.

We start with an auxiliary lemma.
We say that an instance $\langle G, \gamma \rangle$ of {\sc Strip Planarity} on $k>1$ strips is {\em spined} if 
there exists a path $(v_1,\dots,v_k)$ in $G$ such that $\gamma(v_i)=i$, vertex $v_k$ is the unique vertex in the $k$-th strip, and each vertex $v_i$ with $i\neq 1$ induces a component in the $i$-th strip; see Fig.~\ref{fig:left-crossings}. We call path $(v_1,\dots,v_k)$ the {\em spine path} of $\langle G, \gamma \rangle$ and refer to edge $(v_i,v_{i+1})$ as the \emph{$i$-th edge} of such a path.

\begin{lemma}\label{le:spine}
	Any positive spined instance $\langle G, \gamma \rangle$ of {\sc Strip Planarity} admits a strip-planar drawing in which the intersection point between the first edge of the spine path of $\langle G, \gamma \rangle$ and the horizontal line separating the first and the second strip is the left-most intersection point between any inter-strip edge and such a line.
\end{lemma}
\begin{proof}
	Let $\Gamma$ be a strip-planar drawing of $\langle G, \gamma \rangle$; see Fig.~\ref{left-crossing-a}. We show how to construct a strip-planar drawing $\Gamma'$ of $\langle G, \gamma \rangle$ satisfying the requirements of the lemma.
	Consider two horizontal lines $l'$ and $l''$, with $l'$ below $l''$, that lie above any vertex in the first strip and below the horizontal line separating the first and the second strip in $\Gamma$, and that intersect any edge of $G$ at most once.
	Denote by $p_1,\dots,p_m$ (by $q_1,\dots,q_m$) the intersection points between $l'$ (between $l''$) and the edges of $G$ in the left-to-right order along $l'$ (along $l''$).
	Let $\mathcal{R}'$ ($\mathcal{R}''$) be the region delimited by $l'$ (by $l''$) and by the horizontal line separating the second last and the last strip, and lying to the left of the spine path $(v_1,\dots,v_k)$ of $\langle G, \gamma \rangle$ in $\Gamma$.
	
	We obtain $\Gamma'$ as follows; see Fig.~\ref{left-crossing-b}. Initialize $\Gamma'$ to $\Gamma$.
	Remove from $\Gamma'$ the drawing of the part of $G$ in the interior of $\mathcal{R}'$.
	Then, consider the drawing $\Gamma_{\mathcal{R}''}$ of $G$ in the interior of $\mathcal{R}''$ in $\Gamma$. Add to $\Gamma'$ a copy of drawing $\Gamma_{\mathcal{R}''}$ to the right of $\Gamma$, after a horizontal mirroring. Let $q'_i$ be the point of $\Gamma'$ corresponding to the mirrored and translated copy of point $q_i$.
	Finally, complete the drawing of the inter-strip edges crossing lines $l'$ and $l''$ as curves  between points $p_i$ and $q'_i$ in the interior of the first strip in $\Gamma$. The fact that such curves can be drawn in $\Gamma'$ without introducing any crossings and without crossing the horizontal line separating the first and the second strip is due to the fact that points $q'_1,\dots,q'_m$ appear in this right-to-left order along $l''$ and points $p_1,\dots,p_m$ appear in this left-to-right order along $l'$.
\begin{figure}[tb]
	\subfigure[]{\includegraphics[height=.4\columnwidth,page=1]{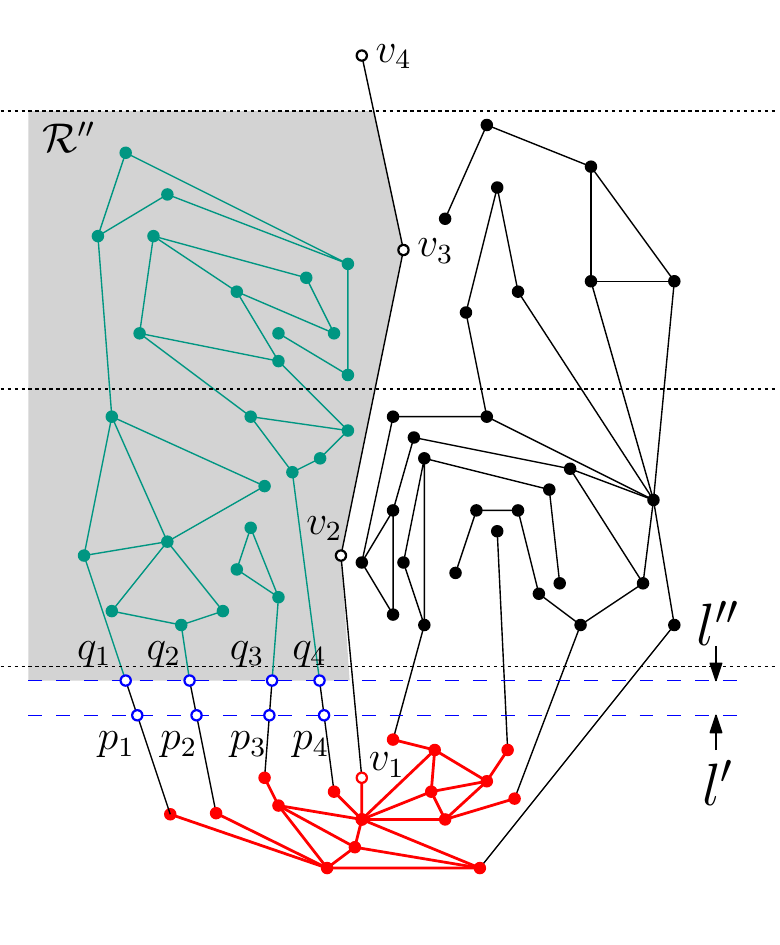}
		\label{left-crossing-a}}
	\hfil
	\subfigure[]{\includegraphics[height=.4\columnwidth, page=2]{spine.pdf}
		\label{left-crossing-b}}
	\caption{Illustrations for the proof of Lemma~\ref{se:strips}.}
	\label{fig:left-crossings}
\end{figure}
\end{proof}


\begin{lemma}\label{le:single-source}
	Let $\langle G=(V,E), \gamma \rangle$ be a spined instance of {\sc Strip Planarity} on $k>1$ strips with a unique source component $c$.
	It is possible to construct in linear time an equivalent spined instance $\langle G'=(V',E'), \gamma' \rangle$ of {\sc Strip Planarity} on $k-1$ strips with a unique source component $c'$.
\end{lemma}

\begin{proof}
	Consider the source component $c$ of $\langle G, \gamma \rangle$, which lies in the first strip.
	First, construct an auxiliary planar graph $G_c$ as follows. Initialize $G_c=c$ and add a dummy vertex $v$ to it. For each intra-strip edge $e$ incident to a vertex $u$ in $c $, add to $G_c$ a dummy vertex $v_e$ and edges $(v,v_e)$ and $(v_e,u)$.
	If $G_c$ contains cut-vertices, then let $B_c$ be the block of $G_c$ that contains $v$.
	Then, construct a PQ-tree $\mathcal{T}_c$ representing all possible orders of the edges around $v$ in a planar embedding of $B_c$. This can be done by applying the planarity testing algorithm of Booth and Lueker~\cite{bt-tcopiggppqa-76}, in such a way that vertex $v$ is the last vertex of the $st$-numbering of block $B_c$. Observe that each leaf of PQ-tree $\mathcal{T}_c$ corresponds to exactly one vertex $v_e$ in $B_c$.
	We construct a {\em representative graph} $G_{\mathcal{T}_c}$ from $\mathcal{T}_c$, as described in~\cite{FengCE95}, composed of (i) {\em wheel} graphs (that is, graphs consisting of a cycle, called {\em rim}, and of a {\em central vertex} connected to every vertex of the rim), of (ii) edges connecting vertices of different rims not creating any simple cycle that contains vertices belonging to more than one wheel, and of (iii) vertices of degree $1$, which are in one-to-one correspondence with the leaves of $\mathcal{T}_c$ (an hence with the dummy vertices $v_e$ in $B_c$), each connected to a vertex of some rim. As proved in~\cite{FengCE95}, in any planar embedding of $G_{\mathcal{T}_c}$ in which all the degree-$1$ vertices are incident to the same face, the order in which such vertices appear in a Eulerian tour of this face is in $O(\mathcal{T}_c)$.
	
	We construct $\langle G', \gamma' \rangle$ as follows. For $i=2,\dots,k$ and for each vertex $v$ such that $\gamma(v)=i$, we add $v$ to $V'$ and we set $\gamma'(v)=i-1$, that is, we assign all the vertices of the $i$-th strip of $\langle G, \gamma \rangle$, with $i \geq 2$, to the $(i-1)$-th strip of $\langle G', \gamma' \rangle$. Further, we add to $E'$ all edges in $E \cap (V' \times V')$.
	Also, we add all the vertices and edges of $G_{\mathcal{T}_c}$ to $V'$ and to $E'$, respectively, and we set $\gamma'(u)=1$, for each vertex $u$ of $G_{\mathcal{T}_c}$. Finally, for each inter-strip edge $e=(x,y)$ in $E$ with $\gamma(x)=1$ and $\gamma(y)=2$, we add to $E'$ an intra-strip edge between vertex $y$ and the degree-$1$ vertex of $G_{\mathcal{T}_c}$ corresponding to $v_e$.
	
	The construction of instance $\langle G', \gamma' \rangle$ can be carried out in linear time since the construction of $\mathcal{T}_c$ takes linear time in the size of $B_c$~\cite{bt-tcopiggppqa-76} and since the construction of
	$G_{\mathcal{T}_c}$ takes linear time in the size of $\mathcal{T}_c$~\cite{FengCE95}. Hence, instance $\langle G', \gamma' \rangle$ has size linear in the size of $\langle G, \gamma \rangle$. 
	Further, instance $\langle G', \gamma' \rangle$ has a unique source component, which contains $G_{\mathcal{T}_c}$ as a subgraph. This is due to the fact that any component in the second strip of $\langle G, \gamma \rangle$ has an inter-strip edge incident to a vertex of $c$. Finally, $\langle G', \gamma' \rangle$ is a spined instance whose spine path is the one obtained from the spine path of $\langle G, \gamma \rangle$ by removing its first edge.
	
	We now show the equivalence between the two instances.

	Suppose that $\langle G, \gamma \rangle$ admits a strip-planar drawing $\Gamma$, we show how to construct a strip-planar drawing $\Gamma'$ of  $\langle G', \gamma' \rangle$. First, observe that all the vertices of $c$ incident to inter-strip edges lie on the outer face of the drawing of $c$ in $\Gamma$. We subdivide each inter-strip edge incident to $c$ with a dummy vertex $v_e$ lying in the interior of the first strip of $\Gamma$. By the construction of $\mathcal{T}_c$ and of $G_{\mathcal{T}_c}$, each degree-$1$ vertex of $G_{\mathcal{T}_c}$ corresponds to exactly one vertex $v_e$.
	Further, let $c^+$ be the subgraph of $G$ induced by the vertices in $c$ and by all the vertices $v_e$.
	Note that the order in which the vertices $v_e$ appear in a Eulerian tour of the outer face of $c^+$ in $\Gamma$ is in $O(\mathcal{T}_c)$. Hence,
	we can replace the drawing of $c^+$ in $\Gamma$ with a drawing of $G_{\mathcal{T}_c}$ in which each degree-$1$ vertex is mapped to the vertex $v_e$ it corresponds to. Finally, we obtain $\Gamma'$ by merging the first two strips of $\Gamma$ into the first strip of $\Gamma'$.
	
	Suppose that $\langle G', \gamma' \rangle$ admits a strip-planar drawing $\Gamma'$; we show how to construct a strip-planar drawing $\Gamma$ of  $\langle G, \gamma \rangle$. First, by Lemma~\ref{le:spine}, we can assume that in $\Gamma'$ 
	the intersection point between the first edge of the spine path of $\langle G', \gamma' \rangle$ and the line separating the first and the second strip in $\Gamma'$ is the left-most intersection point between any edge $(x,y)$ with $\gamma(x)=1$ and $\gamma(y)=2$ and such a line.
	Further, we can assume the following.
	
	\begin{claim}
		\label{cl:1}
		The rim of every wheel $W$ in $G_{\mathcal{T}_c}$ contains in its interior the central vertex of $W$ and no other vertex in~$\Gamma'$. 
	\end{claim}
	
	\begin{proof}
		The claim can be proved with the same techniques used in~\cite{adn-aspbep-15}, by redrawing each edge connecting two adjacent vertices of the rim as a curve arbitrarily close to the length-$2$ path connecting them and passing through the central vertex of the wheel they belong to. This implies that all the degree-$1$ vertices of $G_{\mathcal{T}_c}$ lie in the outer face of the drawing of $G_{\mathcal{T}_c}$ induced by $\Gamma'$.
	\end{proof}
	
	We obtain $\Gamma$ as follows.  We initialize $\Gamma$ as the drawing in $\Gamma'$ of the subinstance of $\langle G', \gamma' \rangle$ induced by the vertices not in $G_{\mathcal{T}_c}$, where the $i$-th strip in $\Gamma'$ is mapped to the $(i+1)$-th strip in $\Gamma$. First, we add a drawing of $G_{\mathcal{T}_c}$ in the first strip of $\Gamma$ that is a copy of the drawing of $G_{\mathcal{T}_c}$ in $\Gamma'$. We now show how to draw in $\Gamma$ the inter-strip edges incident to $G_{\mathcal{T}_c}$. Observe that these edges correspond to the intra-strip edge incident to $G_{\mathcal{T}_c}$ in $\Gamma'$. 
	We draw each inter-strip edge $(x,y)$ with $y$ in $G_{\mathcal{T}_c}$ as a curve composed of six parts. The first part coincides with the drawing of $(x,y)$ in $\Gamma'$; the second part is a curve arbitrarily close to the drawing in $\Gamma'$ of a path in $G_{\mathcal{T}_c}$ from $y$ to the first vertex $v_1$ of the spine path of $\langle G', \gamma' \rangle$; the third part is a curve arbitrarily close to the drawing in $\Gamma'$ of the first edge of the spine path of $\langle G', \gamma' \rangle$ till a point $p$ in the interior of first strip of $\Gamma'$ and arbitrarily close to the boundary of the second strip of $\Gamma'$; the fourth part is a horizontal segment connecting $p$ to a point $q$ lying to the left of $\Gamma'$; the fifth part is a vertical segment connecting $q$ to a point $r$ in the interior of the first strip of $\Gamma$; and, finally, the sixth part is a curve connecting $r$ to $y$. Observe that, by Claim~\ref{cl:1}, all the degree-$1$ vertices of $G_{\mathcal{T}_c}$ lie on its outer face in $\Gamma'$ (and hence in $\Gamma$). Thus, it is possible to draw all the inter-strip edges incident to $G_{\mathcal{T}_c}$ without introducing any crossings, since the curves representing these edges preserve the same containment relationship between vertices and cycles in $\Gamma$ as the corresponding intra-strip edges in $\Gamma'$.
	
	To obtain a strip-planar drawing of $\langle G, \gamma \rangle$ we proceed as follows.
	Let $H$ be the graph obtained from $B_c$ by subdividing each edge $e$ incident to $v$ with a dummy vertex $v_e$ and by removing $v$.
	We replace the drawing of $G_{\mathcal{T}_c}$ in $\Gamma$ with a planar drawing of $H$ such that the vertices $v_e$ appear in a Eulerian tour of its outer face in the same clockwise order as the corresponding degree-$1$ vertices appear in a Eulerian tour of the outer face of $G_{\mathcal{T}_c}$ in $\Gamma$. Recall that these vertices are on the outer face of $G_{\mathcal{T}_c}$ in $\Gamma$, by Claim~\ref{cl:1}.
	Such a drawing of $H$ exists since this order is in $O(\mathcal{T}_c)$~\cite{FengCE95}.
	Finally, to complete $\Gamma$, for each cut-vertex $z$ of $G_c$ separating $B_c$ from a subgraph $G_z$ of $G_c$, we draw graph $G_z$ arbitrarily close to $z$.
	This is possible since none of the vertices of $G_z$, except possibly for $z$, is incident to an inter-strip edge. This concludes the proof of the lemma.
\end{proof}

Let $\langle G, \gamma \rangle$ be an instance of {\sc Strip Planarity} on $k>1$ strips satisfying the properties of Lemma~\ref{le:single-source}. By applying $k-1$ times the transformation of this lemma, we obtain an instance of {\sc Strip Planarity} on $k=1$ strips, that is, an instance whose strip-planarity coincides with the planarity of its underlying graph, which can be tested in linear time~\cite{bt-tcopiggppqa-76}. Hence, we get the following.

\begin{lemma}\label{le:single-source-algo}
	Let $\langle G=(V,E), \gamma \rangle$ be a spined instance of {\sc Strip Planarity} on $k>1$ strips with a unique source component $c$.
	It is possible to decide in $O(k \times n)$ time whether $\langle G, \gamma \rangle$ admits a strip-planar drawing.
\end{lemma}

Given an instance of {\sc Strip Planarity}, one can create $O(n)$ spined instances by attaching the spine path to each of the $O(n)$ vertices in the first strip. Hence, by Lemma~\ref{le:single-source-algo}, we get the following.

\begin{theorem}\label{th:single-source-algo}
	Let $\langle G, \gamma \rangle$ be an instance of {\sc Strip Planarity} on $k$ strips such that there exists a unique source component $c$. 
	It is possible to decide in $O(n^3)$ time whether $\langle G, \gamma \rangle$ admits a strip-planar drawing.
\end{theorem}
\begin{proof}
	Let $v$ be a vertex of $c$.
	We define an instance $I_v = \langle G', \gamma' \rangle$ of  {\sc Strip Planarity} on $k+1$ strips as follows. For each vertex $v$ of $G$, we add vertex $v$ to $V(G')$ and set $\gamma'(v) = \gamma(v)$. Also, we add all the edges in $E(G)$ to $E(G')$. Finally, for $i=2,\dots,k+1$, we add to $G'$ a vertex $v_i$ and set $\gamma(v_i)=i$. Finally, we add edge $(v,v_1)$ and edges $(v_i,v_{i+1})$, for $i=2,\dots,k$. 
	Observe that, by construction, path $(v,v_2,\dots,v_{k+1})$ is such that $v$ belongs to $c$ and each vertex $v_i$, with $2 \leq i \leq k+1$, induces a component in the $i$-th strip. 
	Hence, $I_v$ is a spined instance, which can thus be tested for strip-planarity in $O((k+1)\cdot n)$ time, by Lemma~\ref{le:single-source-algo}.
	
	It is not difficult to see that $\langle G, \gamma \rangle$ admits a strip-planar drawing if and only if there exists at least a vertex $v$ of $c$ that is incident to a vertex $u$ with $\gamma(u)=2$, such that instance $I_v$ admits a strip-planar drawing.
	In fact, the {\em if part} follows from the fact that each $I_v$ contains $\langle G, \gamma \rangle$ as a subinstance. 
	The {\em only if} can be proved by observing that if $\langle G, \gamma \rangle$ admits a strip-planar drawing, then it also admits a strip-planar drawing $\Gamma$ in which there exists a vertex $v$ that is incident to a vertex $u$ with $\gamma(u)=2$ such that the intersection point between edge $(u,v)$ and the line separating the first and the second strip in $\Gamma$ is the left-most intersection point between any edge $(x,y)$ with $\gamma(x)=1$ and $\gamma(y)=2$ and such a line in $\Gamma$.
	
	The time bound descends from that of Lemma~\ref{le:single-source-algo} and from the fact that (i) $k \in O(n)$ and $|I_v| \in O(|\langle G, \gamma \rangle|)$ and that (ii) since $G$ is planar, the number of vertices in $c$ that are incident to a vertex $u$ with $\gamma(u)=2$ is in $O(n)$. 
\end{proof}


\section{C-Planarity with Embedded Pipes}\label{cpp:fixed-embedding}

In this section we show that the {\sc C-Planarity with Embedded Pipes} problem is solvable in quadratic time for a notable family of instances.

Let $c$ be an originating component belonging to a cluster $\mu \in \cal T$ and let $\nu \neq \mu \in \cal T$ be the cluster to which the vertices of $c$ are adjacent to. We say that $c$ is {\em originating from $\mu$ to $\nu$}. 

\begin{lemma}\label{le:cpp-embedded-to-sefe}
	Let $\langle \cgraph{}, \Gamma_A \rangle$ be an instance of {\sc C-Planarity with Embedded Pipes} and let $\cal S$ be the maximum number of originating \me components in a cluster that are incident to the same pipe. It is possible to construct in linear time an equivalent instance \sefeinstance{} of {SEFE} such that (i) $\Gint{}$ is a spanning forest, (ii) each cut-vertex of $\Gb{}=(V,\Eb{})$ is incident to at most one non-trivial block, and (iii) each cut-vertex of $\Gr{}=(V,\Er{})$ is incident to at most $\cal S$ non-trivial blocks.
\end{lemma} 

\begin{proof}
\begin{figure}[tb]
	\subfigure[]{\includegraphics[page=2]{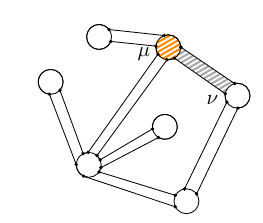}\label{fi:cpp-embedded-a}}
	\hfil
	\subfigure[]{\includegraphics[page=3]{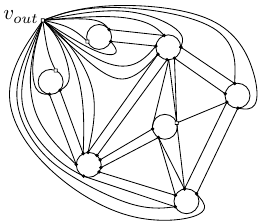}\label{fi:cpp-embedded-b}}
	\hfil
	\subfigure[]{\includegraphics[page=4]{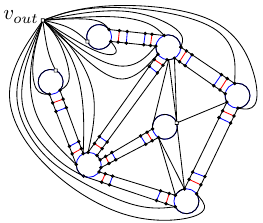}\label{fi:cpp-embedded-c}}
	\caption{(a) Drawing $\Gamma_A$ of the \ga $G_A$ of a flat c-graph; vertices have been placed a the intersection between clusters and pipes. The disk cycle for cluster $\mu$ and the pipe cycle for edge $(\mu,\nu)$ of $G_A$ are depicted as orange and grey tiled regions, respectively. (b) Frame gadget $H$. (c) Partial instance \sefeinstance{} of {\sc SEFE} constructed starting from $\Gamma_A$; graphs $\Gr{}$, $\Gb{}$, and $\Gint{}$ are subdivisions of a triconnected planar graph.}
	\label{fi:cpp-embedded}
\end{figure}
	We show how to construct \sefeinstance{} starting from $\langle \cgraph{}, \Gamma_A \rangle$. The {\em frame gadget} $H$ is an embedded planar graph  defined as follows. 
	For each intersection point between a disk representing a cluster $\mu \in \cal T$ and a segment delimiting a pipe representing an edge of $G_A$ incident to $\mu$ in the drawing $\Gamma_A$ of $G_A$ (see Fig.~\ref{fi:cpp-embedded-a}), we add a vertex at this point. This results in a planar drawing of a graph; we set $H$ to be this graph. We call {\em disk cycle} of $\mu$ the cycle in $H$ obtained from the disk of $\mu$ in $\Gamma_A$. Similarly, we call {\em pipe cycle} of an edge $(\mu,\nu)$ of $G_A$ the cycle in $H$ obtained from the pipe representing edge $(\mu,\nu)$ in $\Gamma_A$. See Fig.~\ref{fi:cpp-embedded-a}.
	Observe that, for each cluster that is incident to exactly one pipe, this operation introduced two copies of the same edge; we subdivide with a dummy vertex the copy that is not incident to the interior of this pipe. Further, we add a vertex $v_{out}$ in the outer face of $H$ and connect it to all the vertices incident to this face. Finally, we triangulate all the faces of $H$ that do not correspond to the interior of any cluster cycle or of any pipe cycle, hence obtaining a triconnected embedded planar graph. See Fig.~\ref{fi:cpp-embedded-b}. 
	
	We initialize $\Gint{} = H$. For each edge $e \in E(H)$ separating the interior of a pipe from the interior of a disk, we remove $e$ from $\Gr{}$ (thus, edge $e$ only belongs to $\Gb{}$). Note that the definition of disk cycles and of pipe cycles can be extended to cycles in $\Gb{}$. Further, for each two edges $e'$ and $e''$ corresponding to the two segments $(u_{\mu,\nu},u_{\nu,\mu})$ and $(v_{\mu,\nu},v_{\nu,\mu})$ delimiting a pipe representing an edge $(\mu,\nu)$ of $G_A$, we subdivide $e'$ with four dummy vertices $a'_{\mu,\nu},b'_{\mu,\nu},b'_{\nu,\mu},a'_{\nu,\mu}$ and $e''$ with four dummy vertices $a''_{\mu,\nu},b''_{\mu,\nu},b''_{\nu,\mu},a''_{\nu,\mu}$, and add edges $\red{(a'_{\mu,\nu},a''_{\mu,\nu})}$ and $\red{(a'_{\nu,\mu},a''_{\nu,\mu})}$ to $\Gr{}$ and edges $\blue{(b'_{\mu,\nu},b''_{\mu,\nu})}$ and $\blue{(b'_{\nu,\mu},b''_{\nu,\mu})}$ to $\Gb{}$.
	
	For each cluster $\mu \in \cal T$, we augment \sefeinstance{} as follows; see Fig.~\ref{fi:cpp-embedded-cluster}.
	We subdivide an edge of $\Gint{}$ that corresponds to a portion of the boundary of the disk representing $\mu$ in $\Gamma_A$ with a dummy vertex $\gamma_\mu$, and we add to $\Gint{}$ a star $C_\mu$, whose central vertex is adjacent to $\gamma_\mu$, having a leaf $z(c_i)$ for each \me component $c_i$ of $\mu$.
	Further, we add to $\Gint{}$ each component $c_i$ of $\mu$. 
	Finally, for each edge $(\mu,\nu)$ of $G_A$, we subdivide edge $(v_{\mu,\nu},a'_{\mu,\nu})$ with a dummy vertex $\alpha_{\mu,\nu}$ and edge $(a''_{\mu,\nu},b''_{\mu,\nu})$ with a dummy vertex $\beta_{\mu,\nu}$. Then, we add to $\Gint{}$ a star $A_{\mu,\nu}$ ($B_{\mu,\nu}$), whose central vertex is adjacent to vertex $\alpha_{\mu,\nu}$ (is identified with vertex $\beta_{\mu,\nu}$), with a leaf $a_\mu(e)$ (a leaf $b_\mu(e)$) for each inter-cluster edge $e$ incident to a component of $\mu$ and to a component in $\nu$.  
	Also, \sefeinstance{} contains the following edges only belonging to $\Gr{}$ and to $\Gb{}$.
	For each inter-cluster edge $e=(x,y)$ with $x \in \mu$ and $y \in \nu$, we add to $\Gr{}$ edges $\red{(x,a_\mu(e))}$, $\red{(y,a_\nu(e))}$, and $\red{(b_\mu(e),b_\nu(e))}$, and we add to $\Gb{}$ edges $\blue{(a_\mu(e),b_\mu(e))}$ and $\blue{(a_\nu(e),b_\nu(e))}$.
	Also, for each vertex $x$ of a component $c_i$ of a cluster $\mu$ such that $x$ is incident to at least an inter-cluster edge, we add to $\Gb{}$ an edge $\blue{(x,z(c_i))}$.
	%
			
	\begin{figure}[tb!]
		\centering
		\includegraphics[page=5]{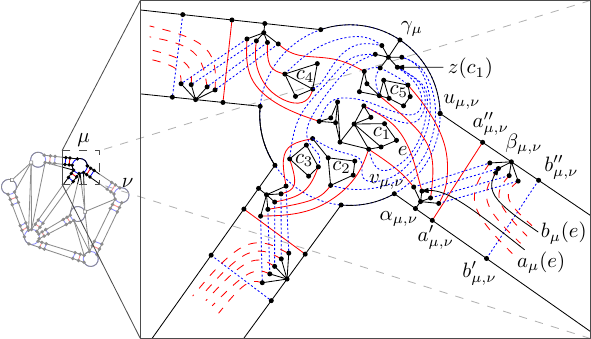}
		\caption{Augmentation of instance \sefeinstance{} focused on cluster $\mu \in \cal T$.}
		\label{fi:cpp-embedded-cluster}
	\end{figure}
	
	Clearly, \sefeinstance{} can be constructed in linear time.
	We now prove that $\Gr{}$ and $\Gb{}$ satisfy the properties of the lemma.
	We note that $\Gr{}$ and $\Gb{}$ are connected, since each vertex of a component $c_i$ is connected to the frame gadget by means of paths in $\Gr{}$ and in $\Gb{}$ passing through stars $A_{\mu,\nu}$ and $C_\mu$, respectively.
	Also, for each cluster $\mu \in \calT{}$, graph $\Gb{}$ contains cut-vertices $\gamma_\mu$, the center of star $C_\mu$, and vertices $z(c_i)$, for each component $c_i$ of $\mu$.
	We now show that all these cut-vertices are incident to at most two non-trivial blocks of $\Gb{}$. 
	Vertex $\gamma_\mu$ is incident to exactly one non-trivial block, that is, the one containing all the vertices and edges of the frame gadget. The center of $C_\mu$ is incident only to non-trivial blocks. Finally, vertices $z(c_i)$, for each component $c_i$ of $\mu$, are incident to at most one non-trivial block, that is, the one containing all the vertices and edges in $c_i$.
	Also, for each cluster $\mu \in \calT{}$, all the passing components in $\mu$ belong to the biconnected component of $\Gr{}$ containing all the vertices and edges of the frame gadget, 
	\begin{figure}[htb]
		\centering
		\includegraphics{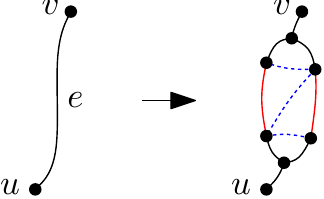}
		\caption{Replacing an edge $e=(u,v)$ to make $\Gint{}$ acyclic.}
		\label{fig:transformation}
	\end{figure}
	while each \me originating component from $\mu$ to a cluster $\nu$ determines a non-trivial block incident to cut-vertex $\alpha_{\mu,\nu}$, and each \se originating component from $\mu$ to a cluster $\nu$ determines a trivial block incident to cut-vertex $\alpha_{\mu,\nu}$. Since the number of \me originating components from any cluster to any other cluster is at most $\cal S$, graph $\Gr{}$ satisfies the required properties. 
	The following claim implies that $\Gint{}$ can be transformed into a spanning forest without altering the properties of \sefeinstance{}.
	
	\begin{claim}\label{cl:cycle-removal}
		Each cycle of $\Gint{}$ can be removed without altering the properties of \sefeinstance{} by replacing one of its edges with the gadget in Fig.~\ref{fig:transformation}.
	\end{claim}
	\begin{proof}
		Let $\phi$ be any cycle in $\Gint{}$. By replacing an edge $e$ of $\phi$ with the gadget described in Figure~\ref{fig:transformation}, we obtain a new instance \sefeinstance{'} of SEFE. Clearly, such a transformation does not introduce any cut-vertex in either $\Gr{}$ or $\Gb{}$. Also, it does not transform any trivial block into a non-trivial block and it does not create any new block, since $e$ used to belong to a cycle in $\Gint{}$.
		We claim that we did not alter the possible vertex-cycle containment relationships of \sefeinstance{}. This is due to the fact that, since $\Gr{}$ and $\Gb{}$ remain connected after removing edge $e$, in any SEFE of \sefeinstance{'} there exists no vertex of $V$ in the interior of the newly introduced cycles.
		Repeating such a replacement until no cycle is left in $\Gint{}$ yields an instance \sefeinstance{*} satisfying the required properties. Since each repetition of the above transformation can be performed in constant time, removing all the cycles can be done in total linear time.
	\end{proof}
	
	We now prove the equivalence.
	
	Suppose that \sefeinstance{} admits a SEFE \sefesolution{}. We show how to construct a c-planar drawing with embedded pipes $\Gamma$ of $\langle \cgraph{}, \Gamma_A \rangle$. Without loss of generality, we assume that vertex $v_{out}$ is embedded on the outer face of \sefesolution{}. Observe that the paths in $\Gint{}$ corresponding to the segments delimiting the pipes representing an edge of $G_A$ incident to a cluster $\mu \in \cal T$ appear in \sefesolution{} in the same clockwise circular order as the corresponding pipes appear around the disk representing cluster $\mu$ in $\Gamma_A$. This is due to the fact that the frame gadget is a triconnected planar graph whose unique planar embedding is the one obtained from $\Gamma_A$. Note that  in \sefesolution{} all the vertices in $V$ appear either in the interior or on the boundary of disk cycles or of pipe cycles. This is due to the fact that removing all the vertices on the boundary of such cycles leaves a connected subgraph of $G_\cup$ and that there exists a unique face of $H$ to which all the vertices belonging to such cycles are incident.

	The proof is based on the fact that any SEFE of \sefeinstance{} has the following properties. 
	\begin{inparaenum}
		\item
		For each cluster $\mu \in \calT{}$, the central vertex of star $C_\mu$ lies in the interior of the disk cycle of $\mu$, and hence all the vertices and edges of the components $c_i$ of $\mu$ lie in the interior of such a cycle, by the connectivity of $\Gb{}$.
		\item 
		For each two clusters $\mu, \nu \in \cal T$, the vertices of the components of $\mu$ and of the components of $\nu$ lie in the interior of different cycles of $\Gr{}$. This is due to the fact that all the components of each cluster $\mu$ are connected by means of paths in $\Gr{}$ to the leaves of a star $A_{\mu,\xi}$, where $\xi$ is a cluster adjacent to $\mu$. Also, all the leaves of these stars lie in the interior of a cycle of $\Gr{}$ delimited by edges of $\Gint{}$ and by edges $\red{(a'_{\mu,\xi_i}, a''_{\mu,\xi_i})}$, for all the clusters $\xi_i$ adjacent to $\mu$.
		\item
		For each inter-cluster edge $e$ connecting a vertex $v$ of a component $c_i$ of cluster $\mu$ to a cluster $\nu$, edge $\red{(v, a_\mu(e))}$ in $\Gr{}$ crosses edge $\blue{(u_{\mu,\nu}, v_{\mu,\nu})}$. This is due to the previous two points and the fact that the leaves of $A_{\mu,\nu}$ lie outside the disk cycle of $\mu$. Note that we can assume that each of these edges crosses edge $\blue{(u_{\mu,\nu}, v_{\mu,\nu})}$ exactly once, as otherwise we could redraw them in such a way to fulfil this requirement.
		\item
		For two adjacent clusters $\mu, \nu \in \cal T$, the order in which the edges in $\Gr{}$ incident to the leaves of $A_{\mu,\nu}$ cross edge $\blue{(u_{\mu,\nu}, v_{\mu,\nu})}$ from $u_{\mu,\nu}$ to $v_{\mu,\nu}$ is the reverse of the order in which the edges in $\Gr{}$ incident to the leaves of $A_{\nu,\mu}$ cross edge $\blue{(u_{\nu,\mu}, v_{\nu,\mu})}$ from $u_{\nu,\mu}$ to $v_{\nu,\mu}$, where the identification between an edge incident to a leaf $a_\mu(e)$ of $A_{\mu,\nu}$ and an edge incident to a leaf $a_\nu(e)$ of $A_{\nu,\mu}$ is based on the inter-cluster edge $e$ they correspond to. This is due to the fact that the order in which the edges in $\Gr{}$ incident to the leaves of $A_{\mu,\nu}$ cross edge $\blue{(u_{\mu,\nu}, v_{\mu,\nu})}$ is transmitted to the leaves of $B_{\mu,\nu}$ via edges in $\Gb{}$ connecting the leaves of $A_{\mu,\nu}$ to the leaves of $B_{\mu,\nu}$, then it is transmitted to the leaves of $B_{\nu,\mu}$ via edges in $\Gr{}$ connecting the leaves of $B_{\mu,\nu}$ to the leaves of $B_{\nu,\mu}$, and finally to the leaves of $A_{\nu,\mu}$ via edges in $\Gb{}$ connecting the leaves of $B_{\nu,\mu}$ to the leaves of $A_{\nu,\mu}$. Note that, all the leaves of these stars lie in the interior of the pipe cycle corresponding to the edge $(\mu,\nu)$ of $G_A$.
	\end{inparaenum}
	
	We describe the correspondence between the SEFE \sefesolution{} of \sefeinstance{} and the c-planar drawing with embedded pipes $\Gamma$ of $\langle \cgraph{}, \Gamma_A \rangle$.
	For each $\mu \in \calT{}$, we draw region $R(\mu)$ as the simple closed region whose boundary coincides with the drawing in $\GammaB{}$ of the disk cycle of $\mu$.
	Each component $c_i$ of a cluster $\mu$ has the same drawing in $\Gamma$ as $c_i$ in \sefesolution{}.
	For each inter-cluster edge $e=(x,y)$ with $x \in \mu$ and $y \in \nu$, the portion of $e$ in the interior of $R(\mu)$ (of $R(\nu)$) coincides with the drawing of edge $\red{(x,a_\mu(e))}$ (of edge $\red{(y,a_\nu(e))}$) between $x$ (between $y$) and the intersection point of this edge with edge $\blue{(u_{\mu,\nu},v_{\mu,\nu})}$ (with edge $\blue{(u_{\nu,\mu},v_{\nu,\mu})}$). To complete the drawing of all the inter-cluster edges between $\mu$ and $\nu$ in the interior of the pipe representing edge $(\mu,\nu)$ in $G_A$, we connect the intersection points between the corresponding edges in $\Gr{}$ and edges $\blue{(u_{\mu,\nu},v_{\mu,\nu})}$ and $\blue{(u_{\nu,\mu},v_{\nu,\mu})}$ by means of a set of non-intersecting curves. This is possible since 
	the order in which the edges in $\Gr{}$ incident to the leaves of $A_{\mu,\nu}$  cross edge $\blue{(u_{\mu,\nu}, v_{\mu,\nu})}$ from $u_{\mu,\nu}$ to $v_{\mu,\nu}$ is the reverse of the order in which the edges in $\Gr{}$ incident to the leaves of $A_{\nu,\mu}$ cross edge $\blue{(u_{\nu,\mu}, v_{\nu,\mu})}$ from $u_{\nu,\mu}$ to $v_{\nu,\mu}$.
	Hence, $\Gamma$ is a c-planar drawing of $\cgraph{}$. The fact that $\Gamma$ can be continuously deformed into a c-planar drawing with embedded pipes of $\langle \cgraph{}, \Gamma_A \rangle$ is due to the fact that the paths in $\Gint{}$ corresponding to the segments delimiting the pipes incident to each cluster $\mu \in \cal T$ appear in \sefesolution{} in the same clockwise order as the corresponding pipes appear around the disk representing $\mu$ in~$\Gamma_A$.
	
	For the other direction, the goal is to construct a SEFE \sefesolution{} of \sefeinstance{} that satisfies all the properties describe above starting from a c-planar drawing with pipes $\Gamma$ of $\langle \cgraph{}, \Gamma_A \rangle$. 
	For each cluster $\mu \in \cal T$,  we draw the disk cycle of $\mu$ as the boundary of the disk of $\mu$ in $\Gamma_A$. Also, for each edge $(\mu,\nu)$ of $G_A$, we draw the corresponding pipe cycle as the boundary of the pipe of edge $(\mu,\nu)$ in $\Gamma_A$. For each cluster $\mu \in \cal T$, each component $c_i$ of $\mu$ has the same drawing in \sefesolution{} as $c_i$ in $\Gamma$. For each edge $(\mu,\nu)$ of $G_A$, the stars $A_{\mu,\nu}$, $B_{\mu,\nu}$, $A_{\nu,\mu}$, and $B_{\nu,\mu}$ are drawn in \sefesolution{} in such a way that the order of their leaves is the same or the reverse of the order in which the inter-cluster edges between $\mu$ and $\nu$ traverse the boundary of the disk of $\mu$ in $\Gamma$. Note that this order is the reverse of the order in which these edges traverse the boundary of the disk of $\nu$ in $\Gamma$. This allows to draw all the edges in $\Gr{}$ and in $\Gb{}$ that are incident to such leaves without introducing any crossings between edges of the same graph. The drawing of star $C_\mu$, for each cluster $\mu \in \cal T$, and of the edges in $\Gb{}$ incident to its leaves can be easily obtained to respect the circular order of the inter-cluster edges incident to each of the components of $\mu$. This concludes the proof of the lemma.
\end{proof}

By Lemma~\ref{le:cpp-embedded-to-sefe} and~Theorem~\ref{th:sefe} we have the following main result.

\begin{theorem}\label{th:cpp-algorithms-2cases}
	{\sc C-Planarity with Embedded Pipes} can be solved in $O(n^2)$ time for instances $\langle \cgraph{}, \Gamma_A \rangle$ such that for each cluster $\mu \in \cal T$ and for each edge $(\mu,\nu)$ in $G_A$ either 
	(\texttt{CASE~1}) cluster $\mu$ contains at most one originating \me component from $\mu$ to $\nu$ 
	or
	(\texttt{CASE~2}) cluster $\mu$ contains at most two \me originating components from $\mu$ to $\nu$ and does not contain any passing component that is incident to $\nu$.
\end{theorem}
\begin{proof}
	Given an instance $\langle \cgraph{}, \Gamma_A \rangle$ of {\sc C-Planarity with Embedded Pipes} by Lemma~\ref{le:cpp-embedded-to-sefe} we can construct in linear time an equivalent instance \sefeinstance{} of {SEFE} (whose size is hence linear in the size of $\cgraph{}$). Also, 
	\sefeinstance{} is such that $\Gint{}$ is a spanning forest, each cut-vertex of $\Gb{}$ is incident to at most one non-trivial block, and each cut-vertex of $\Gr{}$ is incident 
	either to exactly one non-trivial block (\texttt{CASE~1})
	or to at most two non-trivial blocks (\texttt{CASE~2}). 
	Hence, we can apply Theorem~\ref{th:sefe} to decide in $O(|\cgraph{}|^2)$ time whether \sefeinstance{} is a positive instance of {SEFE} (whether $\langle \cgraph{}, \Gamma_A \rangle$ is a positive instance of {\sc C-Planarity with Embedded Pipes}).
\end{proof}

\section{C-Planarity with Pipes}\label{se:cpp-fpt}

In this section we introduce and study the {\sc C-Planarity with Pipes} problem.
A c-planar drawing $\Gamma$ of a flat c-graph $\cgraph{}$ is a {\em c-planar drawing with pipes} of $\cgraph{}$ if, for any two clusters $\mu, \nu \in \calT{}$ that are adjacent in $G_A$ and for any two inter-cluster edges $e_1$ and $e_2$ that are incident to both $\mu$ and $\nu$, one of the two regions delimited by $B(\mu)$, by $B(\nu)$, by $e_1$, and by $e_2$ does not contain any vertex of $G \setminus (\mu \cup \nu)$; two examples are given in Figs.~\ref{fig:example-parameters-2-a} and~\ref{fig:example-parameters-2-b}. The \cpp problem asks for the existence of a c-planar drawing with pipes of a given flat c-graph.

\begin{figure}[tb]
	\centering
	\subfigure[\label{fig:example-parameters-2-a}]{\includegraphics[width=.28\columnwidth,page=1]{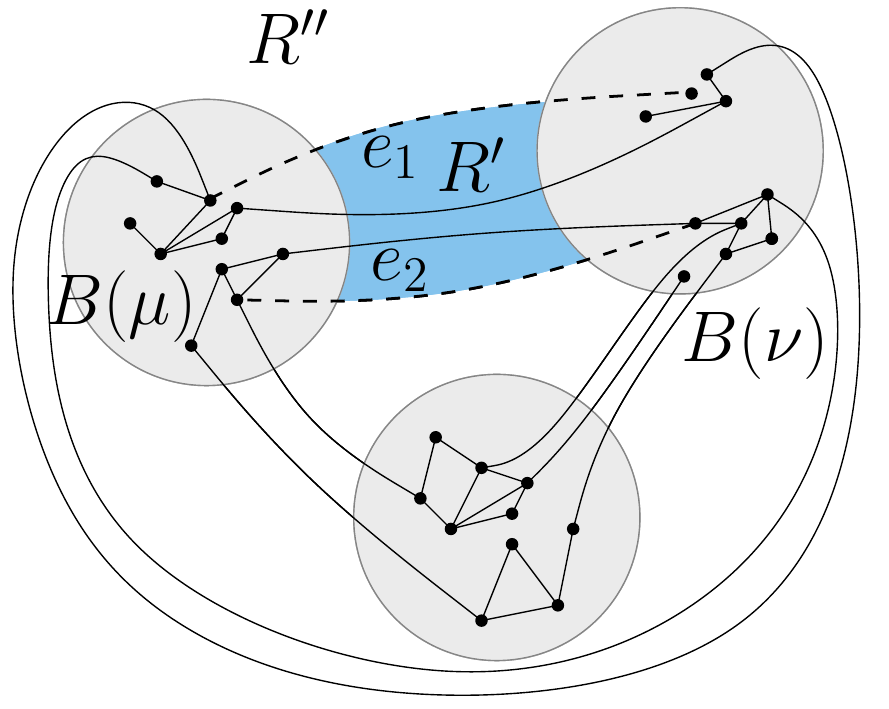}}
	\hfil
	\subfigure[\label{fig:example-parameters-2-b}]{\includegraphics[width=.28\columnwidth,page=2]{cp-with-pipes.pdf}}
	\caption{(a) A c-planar drawing with pipes $\Gamma'$. The two regions $R'$ (blue) and $R''$  delimited by $B(\mu)$, by $B(\nu)$, and by edges $e_1$ and $e_2$ (dashed), where region $R'$ does not contain any vertex of $G \setminus (\mu \cup \nu)$. (b) A c-planar drawing with pipes $\Gamma^*$ corresponding to drawing $\Gamma'$ in which inter-cluster edges are inside pipes.}
	\label{fig:example-parameters-2}
\end{figure}

Note that, if a c-graph \cgraph{} admits a c-planar drawing with pipes, then it is always possible to construct a drawing $\Gamma_A$ of its \ga $G_A$ in which vertices and edges are represented by disks and pipes, respectively, such that $\langle \cgraph{}, \Gamma_A$ is a positive instance of {\sc C-Planarity with Embedded Pipes}; Fig.~\ref{fig:example-parameters-2-b} shows a solution for the instance of {\sc C-Planarity with Embedded Pipes} determined by the c-planar drawing with pipes in Fig.~\ref{fig:example-parameters-2-a}. The following lemma proves that, by suitably augmenting the original c-graph \cgraph{}, it is possible to enforce that the resulting drawing $\Gamma_A$ of $G_A$ respects a specific embedding (if any solution determining a drawing respecting this embedding exists), which implies that {\sc C-Planarity with Pipes} is in fact a generalization of {\sc C-Planarity with Embedded Pipes}.

\begin{lemma}\label{le:cpp-embedded-generalized}
	{\sc C-Planarity with Embedded Pipes} reduces in linear time to {\sc C-Planarity with Pipes}. The reduction does not increase the number of \me components in any cluster.
\end{lemma}
\begin{proof}
	Let $\langle \cgraph{}, \Gamma_A\rangle$ be an instance of {\sc C-Planarity with Embedded Pipes}, where  $\cgraph{}$ is a c-graph and $\Gamma_A$ is planar drawing of $G_A$. We construct an equivalent instance $\cgraph{*}$ of \cpp.
	
	First, we initialize $\cgraph{} = \cgraph{*}$. 
	Then, we augment $\cgraph{*}$ by adding a matching to $G^*$ in such a way that the \ga $G^*_A$ of  $\cgraph{*}$ is a triangulated planar graph.
	In order to do so, we consider a triangulated planar graph $G'_A$ obtained from $G_A$ by adding edges in such a way that the restriction of the unique combinatorial embedding of $G'_A$ to the edges of $G_A$ is the same as the combinatorial embedding of $G_A$ in $\Gamma_A$.
	For each new edge $e=(\mu,\nu)$ of $G'_A \setminus G_A$, we add to $\cgraph{*}$ a new vertex $\mu(e)$ to $\mu$ and a new vertex $\nu(e)$ to $\nu$, and an inter-cluster edge $(\mu(e),\nu(e))$.
	
	Clearly, the reduction can be performed in linear time and $G^*_A$ coincides with $G'_A$. Also, vertices $\mu(e)$ and $\nu(e)$ are \se components of $\mu$ and $\nu$, respectively, and thus the number of \me components remains the same.
	Further, since $G^*_A$ is triconnected, any c-planar drawing with pipes of $\cgraph{*}$ contains a c-planar drawing with pipes of $\cgraph{}$ in which the pipes appear in the desired order around each cluster.
	
	Finally, it is not difficult to see that any c-planar drawing with pipes $\Gamma$ of $\cgraph{}$ in which the order of the pipes incident to each cluster is the same as in $\Gamma_A$ can be extended to a c-planar drawing with pipes of $\cgraph{*}$ by drawing the edges in  $G'_A \setminus G_A$. In fact, for each of these edges $(\mu,\nu)$ there exists a region of $\Gamma$ delimited by a portion of $B(\mu)$ and a portion of $B(\nu)$ where this edge can be drawn, since there exists a face of $\Gamma_A$ to which both $\mu$ and $\nu$ are incident.
\end{proof}

In the remainder of the section we present an FPT algorithm for \cpp in two parameters, namely the maximum number $K$ of \me components in a cluster and the number $c$ of clusters with at least two \me components. Our result is based on a characterization of \cp of flat c-graphs in terms of a newly defined constrained embedding problem. 

\subsection{A Characterization of Flat C-Planarity}\label{se:cp-characterization}

We start with some definitions. Let $\cgraph{}$ be a flat c-graph and let $\mu$ be a cluster in $\calT{}$. A {\em \ctree} $X_\mu$ of $\mu$ is a rooted tree in which every internal vertex is a \me component $c$ of $\mu$ and in which every leaf $x_\mu(e)$ corresponds to an inter-cluster edge $e$ incident to one of such components. A {\em \nctree} $Y_\mu$ of $\mu$ is a rooted tree in which there exists an internal vertex $\nu$ for each cluster $\nu$ adjacent to $\mu$, plus a set of additional internal vertices, and in which every leaf $y_\mu(e)$ corresponds to an inter-cluster edge $e$ incident to $\mu$.
Let $\Gamma$ be a c-planar drawing of $\cgraph{}$, let $X_\mu$ be a \ctree of $\mu$ rooted at a \me component $\rho_\mu$, and let $Y_\mu$ be a \nctree of $\mu$ rooted at a cluster $\xi_\mu$, such that there exists an inter-cluster edge $e_\mu$ incident to both $\rho_\mu$ and $\xi_\mu$.
Let $\calO{}_\mu$ be the clockwise linear order in which the edges incident to $\mu$ traverse $B(\mu)$ in $\Gamma$, starting from and ending at $e_\mu$.
Drawing $\Gamma$ is {\em consistent with} $X_\mu$ if, for each vertex $c \in X_\mu$, the leaves of the subtree of $X_\mu$ rooted at $c$ are consecutive in the restriction of $\calO{}_\mu$ to the inter-cluster edges incident to \me components of $\mu$.
Also, $\Gamma$ is {\em consistent with} $Y_\mu$ if, for each vertex $\nu \in Y_\mu$, the leaves of the subtree of $Y_\mu$ rooted at $\nu$ are consecutive in $\calO{}_\mu$.
Let $\calX{}$ and $\calY{}$ be two sets containing a \ctree $X_\mu$ and a \nctree $Y_\mu$, respectively, for each $\mu$ in $\calT{}$. Drawing $\Gamma$ is {\em consistent with} $\langle \calX{},\calY{} \rangle$ if, for each $\mu \in \calT{}$, drawing $\Gamma$ is consistent with both $X_\mu$ and $Y_\mu$.

Given a flat c-graph $\cgraph{}$, together with two sets $\calX{}$ and $\calY{}$ of \ctrees and of \nctrees, respectively, for all the clusters in $\calT{}$, problem \iccp asks whether a c-planar drawing of $\cgraph{}$ exists that is consistent with $\langle \calX{}, \calY{} \rangle$.

\begin{theorem}\label{th:characterization}
	A flat c-graph $\cgraph{}$ is c-planar if and only if there exist two sets $\calX{}$ and $\calY{}$ of \ctrees and of \nctrees, respectively, for all the clusters in $\calT{}$, such that $\langle \cgraph{}, \calX{}, \calY{}\rangle$ is a positive instance of \iccp.
\end{theorem}

\begin{proof}
	One direction is trivial. Namely, if $\langle \cgraph{},\calX{}, \calY{} \rangle$ is a positive instance of \iccp, then $\cgraph{}$ admits a c-planar drawing (even one that is consistent with $\langle \calX{}, \calY{} \rangle$).
	
	We prove the other direction. Let $\Gamma$ be a c-planar drawing of $\cgraph{}$.
	Consider each cluster $\mu \in \calT{}$. 
	Suppose that there exists at least a \me component $\rho_\mu$ in $\mu$, as otherwise 
	$X_\mu$ is the empty tree and $\Gamma$ is trivially consistent with it.
	Let $e_\mu$ be any inter-cluster edge incident to $\rho_\mu$.
	Let $\calO{}_\mu$ be the clockwise linear order of the edges incident to $\mu$ starting from $e_\mu$ and ending at $e_\mu$. Also, let $\xi_\mu$ be the cluster different from $\mu$ to which $e_\mu$ is incident.
	Since $\Gamma$ is c-planar, there exist no four edges $e_1,e_2,e_3$, and $e_4$ appearing in this order in $\calO{}_\mu$ such that $e_1$ and $e_3$ are incident to a component $c'$ of $\mu$, and $e_2$ and $e_4$ are incident to a component $c''\neq c'$ of $\mu$.  
	Hence, for each two components $c'$ and $c''$ in $\mu$, order $\calO{}_\mu$ defines a unique ``inclusion'' hierarchy with respect to $\rho_\mu$. Namely, we say that $c'$ is {\em nested into} $c''$ if there exists three edges $e_1$, $e_2$, and $e_3$ appearing in this order in $\calO{}_\mu$ such that $e_1$ and $e_3$ are incident to $c''$, and $e_2$ is incident to $c'$. Refer to Fig.~\ref{fig:trees}(a).
	\begin{figure}[tb]
		\centering
		\subfigure[]{\includegraphics[width=.29\columnwidth,page=1]{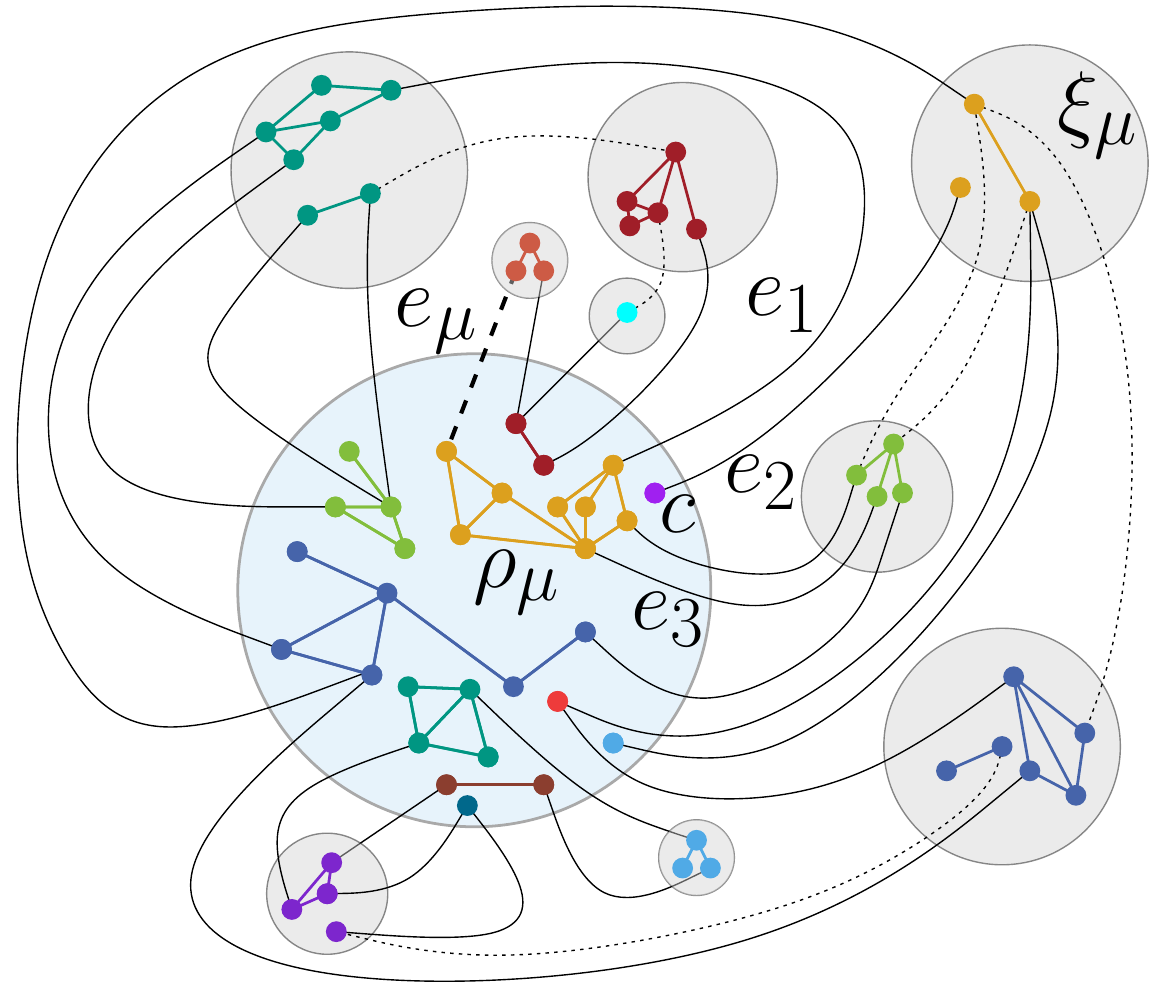}}
		\hfil
		\subfigure[]{\includegraphics[width=.29\columnwidth,page=1]{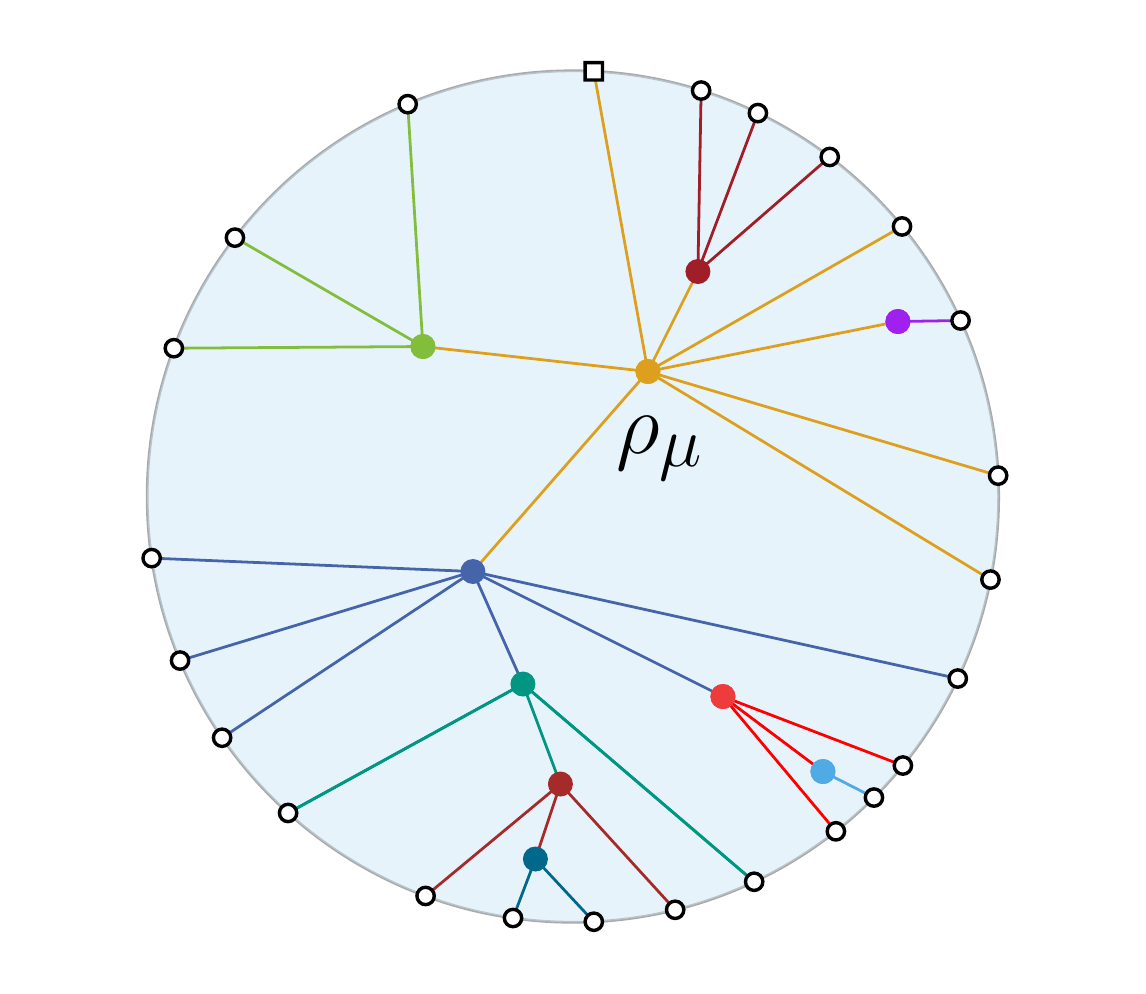}}
		\hfil
		\subfigure[]{\includegraphics[width=.29\columnwidth,page=2]{neighbor-clusters-tree.pdf}}
		\caption{(a) A c-planar drawing $\Gamma$ focused on cluster $\mu$. Edges incident to $\mu$ are solid. Component $c$ is nested into component $\rho_\mu$. Trees (b) $X_\mu$ and (c) $Y_\mu$ such that $\Gamma$ is consistent with $X_\mu$ and $Y_\mu$.}
		\label{fig:trees}
	\end{figure}

	Note that such a hierarchy is acyclic and that every component different from $\rho_\mu$ is nested into $\rho_\mu$, since $\calO{}_\mu$ start and ends at $e_\mu$. 
	We construct a tree $X_\mu$ rooted at $\rho_\mu$ in which every internal vertex is a \me component $c$ of $\mu$ and in which every leaf $x_\mu(e)$ corresponds to an inter-cluster edge $e$ incident to one of such components; refer to Figs.~\ref{fig:trees}(a) and~\ref{fig:trees}(b). There exists an edge $(x_\mu(e),c)$ if and only if edge $e$ is incident to a vertex of component $c \in X_\mu$. Also, there exists an edge $(c',c'')$ if component $c' \in X_\mu$ is nested into component $c'' \in X_\mu$ and there exists no other component $c^*\in X_\mu$ such that $c^*$ is nested into $c''$ and $c'$ is nested into $c^*$ in $\Gamma$.
	By construction, $X_\mu$ is a \ctree{} and $\Gamma$ is consistent with $X_\mu$.
	
	Similarly, order $\calO{}_\mu$ determines whether any two clusters adjacent to $\mu$ in $G_A$ are {\em nested} one into the other; this determines an acyclic hierarchy in which every cluster different from $\xi_\mu$ is nested into $\xi_\mu$. 
	We construct a tree $Y_\mu$ rooted at $\xi_\mu$ in which there exists an internal vertex $\nu$ for each cluster $\nu$ adjacent to $\mu$ in $G_A$ and in which every leaf $y_\mu(e)$ corresponds to an inter-cluster edge $e$ that is incident to $\mu$; refer to Figs.~\ref{fig:trees}(a) and~\ref{fig:trees}(c). There exists an edge $(y_\mu(e),\nu)$ if and only if edge $e$ is incident to a vertex of cluster $\nu \in Y_\mu$. Also, there exists an edge $(\nu',\nu'')$ if cluster $\nu' \in Y_\mu$ is nested into cluster $\nu'' \in Y_\mu$ and there exists no other cluster $\nu^*$ such that $\nu^*$ is nested into $\nu''$ and $\nu'$ is nested into $\nu^*$ in $\Gamma$.
	By construction, $Y_\mu$ is a \nctree{} and $\Gamma$ is consistent with $Y_\mu$.
\end{proof}

In the following theorem, whose proof is deferred to Section~\ref{se:inc-constr-cp}, we show that the \iccp problem can be solved efficiently.

\begin{theorem}\label{th:algo-inc-constr-cp}
	{\sc Inclusion-Constrained C-Planarity} can be solved in quadratic time. 
\end{theorem}

In the following section we prove that, for each cluster $\mu$ of a c-graph \cgraph{}, there exists a unique \nctree $Y_\mu$ such that every c-planar drawing with pipes of \cgraph{} is consistent with $Y_\mu$.
Hence, an FPT algorithm for \cpp can be based on generating, for each cluster, all the possible \ctrees and its unique \nctree, and on testing the corresponding instances of \iccp by Theorem~\ref{th:algo-inc-constr-cp}.

\subsection{Neighbor-clusters Trees in C-Planar Drawings with Pipes}\label{sse:nctrees-cpp}

In the following theorem we give a characterization of the c-graphs that are positive instances of \cpp based on the possible orders of inter-cluster edges around each cluster in any c-planar drawing. We first consider only c-graphs whose \ga $G_A$ has no trivial blocks; however, we prove later that this is not a restriction.

\begin{theorem}\label{th:characterization-cpp}
	Let $\cgraph{}$ be a flat c-graph such that $G_A$ has no trivial block.
	Then, $\cgraph{}$ is a positive instance of \cpp if and only if $\cgraph{}$ admits a c-planar drawing $\Gamma$ in which, for each cluster $\mu \in \calT{}$, the inter-cluster edges between $\mu$ and any cluster $\nu$ adjacent to $\mu$ in $G_A$ are consecutive in the order in which the inter-cluster edges incident to $\mu$ cross $B(\mu)$ in $\Gamma$.
\end{theorem}
\begin{proof}
	One direction is trivial, since any c-planar drawing with pipes of $\cgraph{}$ is a c-planar drawing satisfying the conditions of the theorem.
	
	Suppose that $\cgraph{}$ admits a c-planar drawing $\Gamma$ satisfying the conditions of the theorem. We prove that $\Gamma$ is a c-planar drawing with pipes of $\cgraph{}$. Assume for a contradiction that this is not the case, that is, there exist two clusters $\mu, \nu \in \calT{}$ that are adjacent in $G_A$ and two inter-cluster edges $e_1$ and $e_2$ that are incident to both $\mu$ and $\nu$, such that both the regions delimited by $B(\mu)$, by $B(\nu)$, by $e_1$, and by $e_2$ in $\Gamma$ contain at least a vertex of $G \setminus (\mu \cup \nu)$; see Fig.~\ref{fig:bridge}.
	
	\begin{figure}
		\centering
		\includegraphics{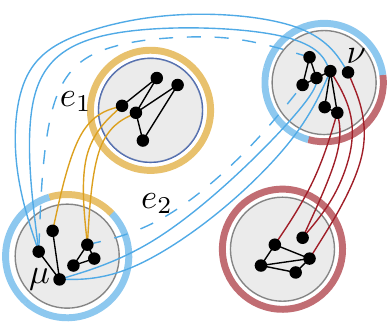}
		\caption{A c-planar drawing that is not a c-planar drawing with pipes, even if the inter-cluster edges incident to the same cluster are consecutive (see the annuli around clusters), due to the presence of trivial block $(\mu,\nu)$.}
		\label{fig:bridge}
	\end{figure}
	
	Note that, if there exists a cluster that is adjacent to $\mu$ (to $\nu$) in $G_A$ in the interior of one of the two regions, then there exists no other cluster that is adjacent to $\mu$ (to $\nu$) in $G_A$ in the interior of the other region, as otherwise the edges between $\mu$ and $\nu$ would not be consecutive around $B(\mu)$ (around $B(\nu)$).
	Hence, for every cluster lying in the interior of one of the regions, all the paths in $G_A$ connecting it to $\mu$ pass through $\nu$; also, for every cluster lying in the interior of the other region, all the paths in $G_A$ connecting it to $\nu$ pass through $\mu$. Therefore, $(\mu,\nu)$ is a trivial block of $G_A$, a contradiction.
\end{proof}

We exploit Theorem~\ref{th:characterization-cpp} to construct a \nctree $Y^\circ_\mu$ of each cluster $\mu \in \calT{}$ such that any c-planar drawing with pipes of $\cgraph{}$ is consistent with $Y^\circ_\mu$.
Tree $Y^\circ_\mu$  is rooted at a vertex $\omega_\mu$. There exists a child $\nu$ of $\omega_\mu$ for each cluster $\nu$ adjacent to $\mu$, having a leaf $y_\mu(e)$ for each inter-cluster edge $e$ incident to $\mu$ and to $\nu$.
We call $Y^\circ_\mu$ the {\em \enctree} of $\mu$.
Theorem~\ref{th:characterization-cpp} and the construction of $Y^\circ_\mu$, for each cluster $\mu \in \calT{}$, imply the following.

\begin{corollary}\label{co:cpp-characterization-y}
	Let $\cgraph{}$ be a c-graph whose \ga has no trivial blocks.
	Then, $\cgraph{}$ 
	admits a c-planar drawing with pipes
	if and only if $\cgraph{}$ admits a c-planar drawing $\Gamma$ in which, for each $\mu \in \calT{}$, drawing $\Gamma$ is consistent with~$Y^\circ_\mu$.
\end{corollary}

Corollary~\ref{co:cpp-characterization-y} allows us to reduce the problem of testing \cpp for a c-graph whose \ga $G_A$ has no trivial blocks to that of testing \iccp, where the role played by the \nctrees is now taken by the \enctrees. 
%
Next, we overcome the requirement that $G_A$ has no trivial block.

\begin{lemma}\label{le:cpp-remove-bridges}
	Let $\cgraph{}$ be an instance of \cpp in which $G_A$ contains trivial blocks. 
	It is possible to construct in linear time an equivalent instance $\cgraph{*}$ of \cpp in which $G^*_A$ has no trivial block. Further, $K_*=K$ and $c_*=c$, where $K$ ($K_*$) is the maximum number of \me components in a cluster of $\cgraph{}$ (of $\cgraph{*}$) and $c$ ($c_*$) is the number of clusters of $\cgraph{}$ (of $\cgraph{*}$) with at least two \me components.
\end{lemma}
\begin{proof}
	Consider any trivial block $(\mu,\nu)$ in $G_A$. We show how to construct an instance $\cgraph{+}$ of \cpp equivalent to $\cgraph{}$ such that (i) the block of $G^+_A$ containing $(\mu,\nu)$ is not a trivial block, (ii) $G^+_A$ does not contain any trivial block that does not belong to $G_A$, and (iii) $K_+=K$ and $c_+=c$, where $K_+$ is the maximum number of \me components in a cluster of $\cgraph{+}$ and $c_+$ is the number of clusters of $\cgraph{+}$ with at least two \me components. Repeating such a transformation eventually yields an instance $\cgraph{*}$ satisfying the required properties.
	
	We initialize $\cgraph{+} = \cgraph{}$. Then, we add a new cluster $\eta$ to $\calT{+}$, which only contains a new vertex $v$. Also, we add a vertex $u_\mu$ to cluster $\mu$ and a vertex $u_\nu$ to cluster $\nu$, and edges $(v,u_\mu)$ and $(v,u_\nu)$ to $\cgraph{+}$. 
	
	We prove that $\cgraph{+}$ and $\cgraph{}$ are equivalent. 
	One direction is trivial, as any c-planar drawing with pipes of $\cgraph{+}$ contains a c-planar drawing with pipes of $\cgraph{}$.
	
	Suppose that $\cgraph{}$ admits a c-planar drawing with pipes $\Gamma$. Consider the two inter-cluster edges $e_1$ and $e_2$ adjacent to both $\mu$ and $\nu$ such that the region $R_\mu$ delimited by $B(\mu)$, by $B(\nu)$, by $e_1$, and by $e_2$ containing all the vertices of $G \setminus (\mu \cup \nu)$ does not contain any other inter-cluster edge adjacent to both $\mu$ and $\nu$ in $\Gamma$. We construct a c-planar drawing with pipes $\Gamma^+$ of $\cgraph{+}$ starting from $\Gamma$. Namely, draw path $(u_\mu,v,u_\nu)$ as a curve arbitrarily close to edge $e_1$ in $\Gamma$ in the interior of region $R_\mu$ introducing neither edge-edge nor edge-region crossings, and draw $R(\eta)$ as a simple closed region enclosing only the vertex $v$.
	
	The time bound descends from the fact that each augmentation step described above can be performed in constant time and that the number of trivial blocks in $G_A$ is at most linear in the size of $\cgraph{}$.
	
	Finally, $K_+=K$ and $c_+=c$, since $u_\mu$ and $u_\nu$ are \se components of $\mu$ and of $\nu$, respectively, while $\eta$ contains exactly one component, which is a \me component. 
\end{proof}

\subsection{An FPT Algorithm for C-Planarity with Pipes}\label{se:fpt-cpp}
In the following we prove the main result of the section.

\begin{theorem}\label{th:fpt-cplanarity-pipes}
	\cpp can be tested in $O(K^{c(K-2)}) \cdot O(n^2)$ time, where
	$K$ is the maximum number of \me components in a cluster 
	and $c$ is the number of clusters with at least two \me components.
\end{theorem}
\begin{proof}
	Let $\cgraph{}$ be an instance of \cpp. 
	First, apply Lemma~\ref{le:cpp-remove-bridges} to construct in linear time an equivalent instance $\cgraph{*}$ of \cpp whose \ga contains no trivial blocks (possibly $\cgraph{*} = \cgraph{}$) and such that $K_*=K$ and $c_*=c$, where
	$K_*$ is the maximum number of \me components in a cluster of $\cgraph{*}$ 
	and $c_*$ is the number of clusters of $\cgraph{*}$ with at least two \me components.
	Second, construct the set $\calY{}$ containing the unique \enctree $Y^\circ_\mu$, for each cluster $\mu \in \calT{*}$. 
	Then, construct all the possible sets $\calX{}$ of \ctrees, for each cluster $\mu \in \calT{*}$, as follows. 
	If $\mu$ does not contain any \me component, then this set contains only the empty tree, while if $\mu$ contains exactly one \me component $c$, then this set contains only a star whose central vertex is $c$, with a leaf $x_\mu(e)$ for each inter-cluster edge $e$ incident to $c$. Otherwise, consider a set $\calI{}$ containing a vertex $c$ for each \me component $c$ of $\mu$. 
	We generate all the trees on the vertices in $\calI{}$ and, for each of them, we add to each vertex $c$ 
	a leaf $x_\mu(e)$ for each inter-cluster edge $e$ incident to $c$; by Cayley's formula~\cite{az-pftb-04}, the number of these trees is $k_\mu^{k_\mu-2}$, where $k_\mu$ is the number of \me components of $\mu$.
	Finally, apply Theorem~\ref{th:algo-inc-constr-cp} to test whether
	$\langle  \cgraph{*}, \calX{}, \calY{} \rangle$ is a positive instance of \iccp, for each pair $\langle \calX{}, \calY{} \rangle$.
	By Theorem~\ref{th:characterization-cpp} and~Corollary~\ref{co:cpp-characterization-y}, we conclude that $\cgraph{*}$ is a positive instance of \cpp  if and only if at least one of such tests succeeds.
	
	There exist $\Pi_{\mu \in {\cal S}}   k_\mu^{k_\mu-2}$ combinations of \ctrees over all clusters in $\calT{}$, where $\cal S$ is the set of clusters in $\cal T$ containing at least two \me components, which we can upper bound by $K^{c(K-2)}$, where $K$ is the maximum number of \me components in a cluster and $c=|{\cal S}|$. Since there exists a unique set $\calY{}$ of \enctrees for $\cgraph{}$ and since each application of Theorem~\ref{th:algo-inc-constr-cp} requires quadratic time, the statement follows.
\end{proof}

We observe two notable corollaries of Theorem~\ref{th:fpt-cplanarity-pipes} (for the second, see Lemma~\ref{le:cpp-embedded-generalized}).

\begin{corollary}\label{co:strip}
	{\sc Strip Planarity} can be tested in $O(K^{s(K-2)})\cdot O(n^2)$ time, where
	$K$ is the maximum number of \me components in a strip and              
	$s$ is the number of strips containing at least two \me components.
\end{corollary}

\remove{
	It is worth pointing out the importance of setting the parameter $K$ of the FPT algorithm as the maximum number of \me components in a strip, not counting the \se components. In fact, since in an instance of {\sc Strip Planarity} all the blocks of $G_A$ are trivial blocks, the \se components added in each strip to remove such trivial blocks would not have allowed us to infer Corollary~\ref{co:strip}.
}

\begin{corollary}\label{th:algo-pipes-embedded}
	{\sc C-Planarity with Embedded Pipes} can be tested in $K^{c(K-2)} \cdot O(n^2)$ time, where
	$K$ is the maximum number of \me components in a cluster and $c$ is the number of clusters with at least two \me components.
\end{corollary}

\section{Proof of Theorem~\ref{th:algo-inc-constr-cp}}\label{se:inc-constr-cp}

In this section we give a proof of Theorem~\ref{th:algo-inc-constr-cp}, which has been stated in Section~\ref{se:cp-characterization}, by describing an algorithm that is based on a linear-time reduction (Lemma~\ref{le:inc-constr-cp-2-SEFE}) from instances of \iccp to equivalent instances of SEFE that can be solved in quadratic time by Theorem~\ref{th:sefe}.
We first describe the reduction in Lemma~\ref{le:inc-constr-cp-2-SEFE} and then discuss its implications to complete the proof of Theorem~\ref{th:algo-inc-constr-cp}.

\begin{lemma}\label{le:inc-constr-cp-2-SEFE}
	Let $\langle \calC{}(G,\calT{}), \calX{}, \calY{} \rangle$ be an instance of \iccp. 
	It is possible to construct in linear time an equivalent instance $\langle \Gr{}(V,\Er{}),\Gb{}(V,\Eb{})\rangle$ of SEFE in which the common graph $G_\cap=(V,\Er{} \cap \Eb{})$ is a forest such that the cut-vertices of \Gr{} and \Gb{} are incident to at most two non-trivial blocks.
\end{lemma}

\begin{proof}
	For each cluster $\mu \in \calT{}$ instance \sefeinstance{} contains a {\em cluster gadget} $G_\mu$ composed of edges in $\Er{} \cup \Eb{}$. 
	These gadgets are then attached by means of edges in $\Eb{}$ to an outer {\em frame}, composed of edges of $\Gint{}$, which enforces them to lie ``outside of each other''.
	Finally, these gadgets are connected with each other by means of edges in $\Er{}$ representing inter-cluster edges.
	
	Our reduction is inspired by the original reduction from {\sc C-Planarity} to SEFE proposed by Schaefer~\cite{s-ttphtpv-13}. However, while that reduction produces instances of SEFE in which the cut-vertices of $\Gr{}$ and $\Gb{}$ may have a linear number of non-trivial blocks, we exploit the presence of the components-tree and of the neighbor-clusters-tree to create instances in which the non-trivial blocks incident to cut-vertices are at most two, which makes instance \sefeinstance{} polynomial-time solvable.
	We now describe in detail the construction of $\Gr{}$ and $\Gb{}$.
	
	For each cluster $\mu \in \calT{}$, cluster gadget $G_\mu$ is constructed as follows. Refer to Fig.~\ref{fig:gadget}.
	\begin{figure}[tb!]
		\centering
		\includegraphics[height=.47\textwidth,page=1]{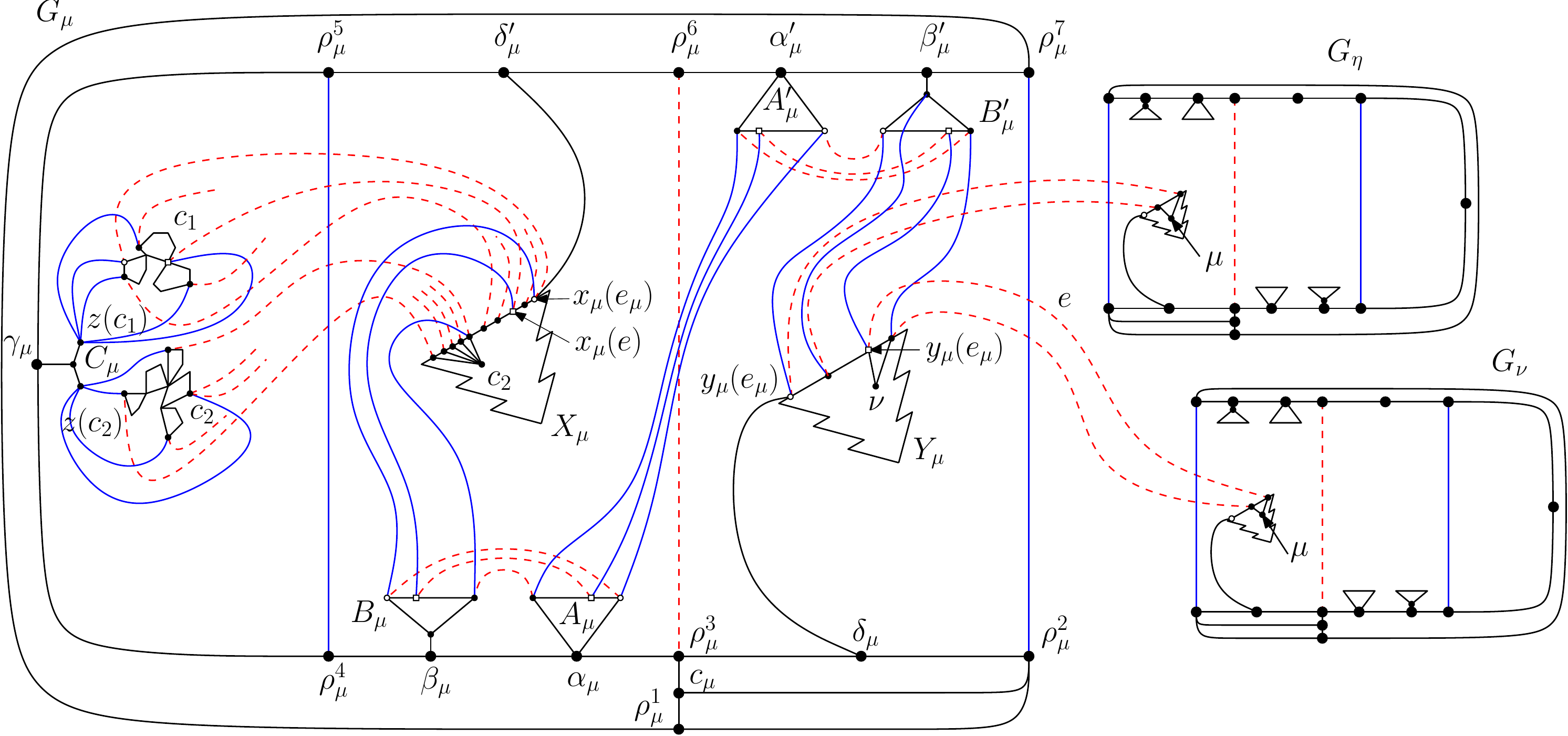}
		\caption{
			Sketch of the cluster gadgets $G_\mu$, $G_\nu$, and $G_\eta$ for cluster $\mu$ and its neighbors $\nu$ and $\eta$. For readability purposes, edges of $\Gint{}$ in $G_\mu$ between the center $c_\mu$ of the wheel $W_\mu$ and some vertices of the rim of $W_\mu$ have been omitted.}
		\label{fig:gadget}
	\end{figure}

	We first describe the part of $G_\mu$ that belongs to both $\Gr{}$ and $\Gb{}$.
	Gadget $G_\mu$ contains a wheel $W_\mu$ with a {\em central vertex} $c_\mu$ that is connected to all the vertices of a cycle $(\rho^1_\mu,\rho^2_\mu,\delta_\mu,\rho^3_\mu,\alpha_\mu,$ $\beta_\mu,\rho^4_\mu,\gamma_\mu,\rho^5_\mu,\delta'_\mu,\rho^6_\mu, \alpha'_\mu,\beta'_\mu, \rho^7_\mu)$, which is the {\em rim} of $W_\mu$. 
	Also, it contains a star $A_\mu$ ($A'_\mu$), centered at $\alpha_\mu$ (at $\alpha'_\mu$), with a leaf $a_\mu(e)$ (a leaf $a'_\mu(e)$) for each inter-cluster edge $e$ incident to $\mu$ that is incident to a \me component of $\mu$.
	Then, $G_\mu$ contains a star $B_\mu$ ($B'_\mu$), whose central vertex is adjacent to vertex $\beta_\mu$ (to vertex $\beta'_\mu$), with a leaf $b_\mu(e)$ (a leaf $b'_\mu(e)$) for each inter-cluster edge $e$ incident to a \me component of $\mu$.
	Further, it contains a star $C_\mu$, whose central vertex is adjacent to vertex $\gamma_\mu$, with a leaf $z(c_i)$ for each \me component $c_i$ of $\mu$.
	Additionally, $G_\mu$ contains a copy of each \me component $c_i$ of $\mu$.
	Gadget $G_\mu$ also contains trees $X_\mu \in \calX{}$ and $Y_\mu \in \calY{}$; recall that,
	$X_\mu$ has a leaf $x_\mu(e)$ for each inter-cluster edge $e$ incident to a \me component of $\mu$, while $Y_\mu$ has a leaf $y_\mu(e)$ for each inter-cluster edge $e$ incident to $\mu$.
	Finally, $G_\mu$ contains an edge $(y_\mu(e_\mu),\delta_\mu)$ and an edge $(x_\mu(e_\mu),\delta'_\mu)$, where $e_\mu$ is an arbitrary inter-cluster edge incident to the root $\rho_\mu$ of $X_\mu$, if $X_\mu$ is not the empty tree, or an arbitrary inter-cluster edge incident to $\mu$, otherwise.
	
	We now describe the edges of $G_\mu$ only belonging to $\Er{}$.
	Namely, $\Er{}$ contains an edge $\red{(\rho^3_\mu,\rho^6_\mu)}$. Also, for each inter-cluster edge $e$ incident to a vertex $v$ belonging to a \me component of $\mu$, set $\Er{}$ contains an edge $\red{(v,x_\mu(e))}$, an edge $\red{(b_\mu(e),a_\mu(e))}$, and an edge $\red{(a'_\mu(e),b'_\mu(e))}$.
	
	Finally, we describe the edges of $G_\mu$ only belonging to $\Eb{}$. 
	Namely, $\Eb{}$ contains edges $\blue{(\rho^2_\mu,\rho^7_\mu)}$ and $\blue{(\rho^4_\mu,\rho^5_\mu)}$.
	Also, for each vertex $v$ of a \me component $c_i$ of $\mu$ such that $v$ is incident to at least an inter-cluster edge, set $\Eb{}$ contains an edge $\blue{(z(c_i),v)}$.
	Further, for each inter-cluster edge $e$ incident to a \me component of $\mu$, set $\Eb{}$ contains an edge $\blue{(x_\mu(e),b_\mu(e))}$, an edge $\blue{(a_\mu(e),a'_\mu(e))}$, and an edge $\blue{(b'_\mu(e),y_\mu(e))}$.
	Finally, for each inter-cluster edge $e$ incident to a \se component of $\mu$, set $\Eb{}$ contains an edge connecting $y_\mu(e)$ with the center of star $B'_\mu$.
	This concludes the construction of $G_\mu$.
	
	We then add to \Gint{} a {\em frame} consisting of cycle $C=(\sigma_{\mu_1}, \dots, \sigma_{\mu_k}, \sigma^*)$, with $\mu_i \in \calT{}$. Also, we add to $\Er{}$ an edge $\red{(\sigma^*,\rho^1_{\mu_1})}$. Finally, we add to \Eb{} an edge $\blue{(\rho^1_{\mu_i},\sigma_{\mu_i})}$ for each cluster $\mu_i \in \calT{}$. Refer to Fig.~\ref{fig:gadgets-composition-same-side}.
	\begin{figure}[t!]
		\centering
		\includegraphics[width=0.7\columnwidth,page=5]{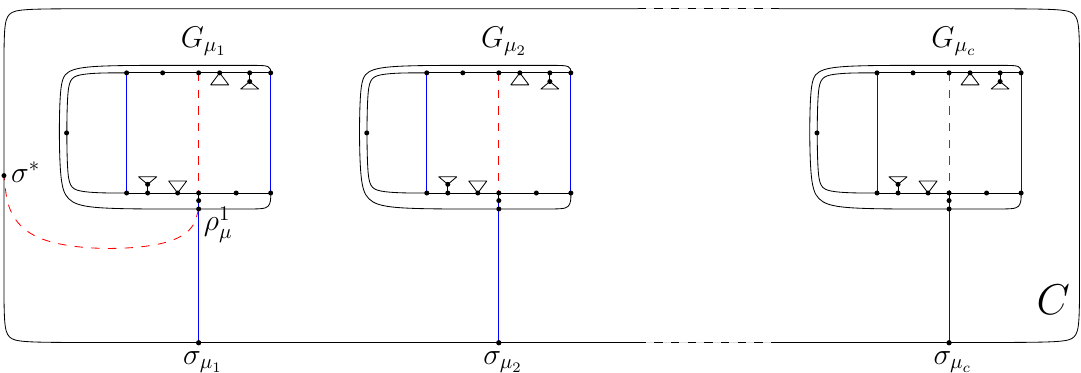}
		\caption{Composing all the cluster gadgets so that they lie in the same side of the frame cycle $C$.}
		\label{fig:gadgets-composition-same-side}
	\end{figure}
	
	To complete the construction of \sefeinstance{}, for each inter-cluster edge $e$ we add to $\Er{}$ an edge $\red{(y_\mu(e),y_\nu(e))}$, where $\mu$ and $\nu$ are the clusters edge $e$ is incident to.
	
	Clearly, \sefeinstance{} can be constructed in linear time and, hence, its size is linear in the size of $\langle \cgraph{}, \calX{}, \calY{} \rangle$. We now prove the equivalence.
	
	$(\Longrightarrow)$ Suppose that \sefeinstance{} admits a SEFE \sefesolution{}. We show how to construct a c-planar drawing $\Gamma$ of $\cgraph{}$ that is consistent with $\langle \calX{}, \calY{} \rangle$.
	
	In the following we will assume that the frame cycle $C$ bounds the outer face of both $\GammaR{}$ and $\GammaB{}$. This is not a loss of generality; in fact, since $G_\cup \setminus C$ is connected, 
	where $G_\cup=(V,E_1 \cup E_2)$,
	all the vertices of \Gr{} and \Gb{} not in $C$ lie on the same side of $C$ in \sefesolution{}, thus $C$ delimits a face in both $\GammaR{}$ and $\GammaB{}$, which we can assume to be the outer face.
	
	We now prove a set of properties of \sefesolution{} with respect to the {\em vertex-cycle containment} relationship, that is, we prove that certain vertices have to lie in the interior or in the exterior of certain cycles of $\Gr{}$ or $\Gb{}$ in $\GammaR{}$ or $\GammaB{}$, respectively. 
	
	We first focus on vertices and cycles belonging to the same cluster gadget $G_\mu$.
	
	For them, we first prove relationships involving cycles belonging to $\Gint{}$, which hence hold in both $\GammaR{}$ and $\GammaB{}$.
	First, the center $c_\mu$ of the wheel $W_\mu$ of $G_\mu$ lies in the interior of the rim of $W_\mu$ in both $\GammaR{}$ and $\GammaB{}$, since vertex $\rho^1_\mu$ is connected in $\Gb{}$ with vertex $\sigma_\mu$ of $C$, which delimits the outer face of both $\GammaR{}$ and $\GammaB{}$, by assumption.
	Also, since the subgraph of $G_\cup$ induced by the vertices in $G_\mu \setminus W_\mu$ is connected, all the vertices of $G_\mu \setminus W_\mu$ lie in the exterior of the rim of $W_\mu$ in both $\GammaR{}$ and $\GammaB{}$.
	
	We then describe further relationships in $G_\mu$ only holding in $\GammaB{}$:
	\begin{inparaenum}[(i)]
		\item
		all the vertices of the copies of the components $c_i$ of $\mu$ lie in the interior of cycle $(\rho^4_\mu,\rho^5_\mu,\gamma_\mu)$ in $\GammaB{}$, since they are all connected to $\gamma_\mu$ by means of paths in $\Gb{}$, and since they cannot lie in the interior of $W_\mu$; and
		\item all the other vertices of $G_\mu$ lie in the interior or on the boundary of cycle $(\rho^4_\mu,\beta_\mu,\alpha_\mu,\rho^3_\mu,\delta_\mu,\rho^2_\mu,\rho^7_\mu,\beta'_\mu,\alpha'_\mu,\rho^6_\mu,\delta'_\mu,\rho^5_\mu)$ in $\GammaB{}$, since they are all connected to vertices $\alpha_\mu$, $\alpha'_\mu$, $\beta_\mu$, and $\beta'_\mu$ by means of paths in $\Gb{}$, and since they cannot lie in the interior of $W_\mu$.
	\end{inparaenum}
	
	Finally, we describe analogous relationships in $G_\mu$ only holding in $\GammaR{}$:
	\begin{inparaenum}[(i)]
		\item all the vertices of the copies of the components $c_i$ of $\mu$, all the vertices of tree $X_\mu$, and all the vertices of stars $A_\mu$ and $B_\mu$ lie in the interior or on the boundary of cycle $(\rho^4_\mu, \beta_\mu, \alpha_\mu, \rho^3_\mu, \rho^6_\mu,\delta'_\mu,\rho^5_\mu,\gamma_\mu)$ in $\GammaR{}$, since they are all connected to vertices $\delta'_\mu$, $\alpha_\mu$, and $\beta_\mu$ by means of paths in $\Gr{}$, and since they cannot lie in the interior of $W_\mu$;
		\item all the other vertices of $G_\mu$ lie in the exterior or on the boundary of cycle 
		$(\rho^2_\mu,\delta_\mu,\rho^3_\mu,\rho^6_\mu,\alpha'_\mu,\beta'_\mu,\rho^7_\mu,\rho^1_\mu)$ in $\GammaR{}$, since they are all connected to vertices $\alpha'_\mu$, $\beta'_\mu$, and $\delta_\mu$ in $\Gr{}$, and since they cannot lie in the interior of $W_\mu$.
	\end{inparaenum}
	
	We now consider vertex-cycle containment relationships between vertices not belonging to $G_\mu$ and cycles in $G_\mu$. In particular, we prove that no vertex $v \notin G_\mu$ lies in the interior of a cycle of $G_\mu$ in both $\GammaR{}$ and $\GammaB{}$.
	Namely, consider any vertex $v \in \nu$, where $\nu \neq \mu$ is a cluster of $\calT{}$. Vertex $v$ does not lie in the interior of cycle $\blue{C^2_\mu}=(\rho^2_\mu,\rho^7_\mu,\rho^1_\mu)$ in \sefesolution{}, due to the fact that $\blue{C^2_\mu}$ is composed of edges belonging to $\Gb{}$, to the fact that $\rho^1_\nu$ does not lie in the interior of $\blue{C^2_\mu}$ (since it is connected to $\sigma_\nu$ in $\Gb{}$, which is incident to the outer face), and to the fact that there exists a path in $\Gb{}$ between $v$ and $\rho^1_\nu$ that does not contain vertices of $\blue{C^2_\mu}$.
	Analogously, vertex $v$ does not lie in the interior of cycle $\red{C^1_\mu}=(\rho^2_\mu,\delta_\mu,\rho^3_\mu,\rho^6_\mu,\alpha'_\mu,\beta'_\mu, \rho^7_\mu,\rho^1_\mu)$ in \sefesolution{}. In fact, there exists a path in $\Gr{}$ not containing vertices of $\red{C^1_\mu}$ between $v$ and a leaf of $Y_\mu$, which lies in the exterior of $\red{C^1_\mu}$. This path exists since the \ga $G_A$ is connected.
	
	We remark that the latter consideration allows us to assume that edge $\red{(\rho^3_\mu,\rho^6_\mu)}$ does not cross edge $\blue{(\rho^2_\mu,\rho^7_\mu)}$ in \sefesolution{}. In fact, in this case we could remove such crossings by redrawing the portions of $\red{(\rho^3_\mu,\rho^6_\mu)}$ lying outside cycle $\blue{C^2_\mu}$ so that edge $\red{(\rho^3_\mu,\rho^6_\mu)}$ is drawn entirely inside $\blue{C^2_\mu}$, without changing the vertex-cycle containment relationships between any vertex and cycle $\red{C^1_\mu}$.
	This implies that no edge of $\Gr{}$ either between two vertices of $G_\mu$ or between two vertices of $G_\nu$, for any cluster $\nu \neq \mu$, crosses edge $\blue{(\rho^2_\mu,\rho^7_\mu)}$.
	In fact, any edge of $\Gr{}$ between two vertices of $G_\mu$ lies entirely in the interior of cycle $\red{C^1_\mu}$, and thus of cycle $\blue{C^2_\mu}$, while any edge of $\Gr{}$ between two vertices of $G_\nu$ lies in the interior of cycle $\red{C^1_\nu}$, thus of cycle $\blue{C^2_\nu}$, and hence entirely in the exterior of~$\blue{C^2_\mu}$.
	
	We now show that we can further assume that all the other edges of $\Gr{}$ cross edge $\blue{(\rho^2_\mu,\rho^7_\mu)}$ at most once. Note that the only remaining edges of $\Gr{}$ that we have to consider are $\red{(\sigma^*,\rho^1_{\mu_1})}$ and all the edges $\red{(y_\nu(e),y_\eta(e))}$ between a leaf of $Y_\nu$ in $G_\nu$ and a leaf of $Y_\eta$ in $G_\eta$, where $\nu, \eta \in \calT{}$ (possibly $\nu=\mu$).
	Namely, from the vertex-cycle containment relationships we proved above it follows that none of the vertices not belonging to $G_\mu$ lies inside cycle $\blue{C^2_\mu}$, and that the only vertices of $G_\mu$ lying in the interior of $\blue{C^2_\mu}$ and not in the interior of $\red{C^1_\mu}$ are the vertices of $A'_\mu$, of $B'_\mu$, and of $Y_\mu$.
	Hence, if an edge $\red{e_r}$ of $\Gr{}$ crosses edge $\blue{(\rho^2_\mu,\rho^7_\mu)}$ more than once, these vertices are the only ones that might be enclosed in a region delimited by $\red{e_r}$ and by $\blue{(\rho^2_\mu,\rho^7_\mu)}$. However, this is not possible since all of them are connected to vertices of $W_\mu$ (namely $\alpha'_\mu$, $\beta'_\mu$, and $\delta_\mu$) by means of paths of edges belonging to $\Gr{}$. Hence, any of these regions does not contain any vertex, and thus we can redraw edge $\red{e_r}$ so that it crosses $\blue{(\rho^2_\mu,\rho^7_\mu)}$ at most once without changing the vertex-cycle containment relationships between any vertex and any cycle in $\Gr{}$.
	
	In the following we will hence assume that, for every cluster $\mu \in \calT{}$, edge $\blue{(\rho^2_\mu,\rho^7_\mu)}$ is crossed at most once by any edge of $\Gr{}$. In particular, this edge is crossed only by each edge $\red{(y_\mu(e),y_\eta(e))}$, incident to a leaf of tree $Y_\mu$, which corresponds to an inter-cluster edge $e$ of $\cgraph{}$ incident to $\mu$.
	
	We now show how to construct $\Gamma$. We denote by $\Theta(T_\mu)$, for each tree $T_\mu \in \{A_\mu,B_\mu,A'_\mu,B'_\mu,X_\mu,Y_\mu\}$, the 
	order of the leaves in $T$ in a clockwise Eulerian tour of $T_\mu$ starting from the leaf corresponding to $e_\mu$ in \sefesolution{}.
	Further, we denote by $\Phi(Y_\mu)$ the order $\Theta(Y_\mu)$ restricted to the leaves corresponding to edges that are incident to \me components of $\mu$.
	Also, we will denote by $\overline{\Theta(T_\mu)}$ the reverse of order $\Theta(T_\mu)$ and by $\overline{\Phi(Y_\mu)}$ the reverse of order $\Phi(Y_\mu)$.
	
	For each cluster $\mu \in \calT{}$, the drawing of each \me component $c_i$ of $\mu$ in $\Gamma$ coincides with the drawing in \sefesolution{} of the copy of $c_i$ in gadget $G_\mu$, which belongs to $\Gint{}$.
	Also, the boundary $B(\mu)$ of the region $R(\mu)$ representing cluster $\mu$ coincides with the drawing of cycle $\blue{C^2_\mu}$ in $\GammaB{}$.
	
	We show how to draw the inter-cluster edges of $\cgraph{}$. 
	In order to do that, we first construct a set $\Lambda_\mu$ of curves for each cluster $\mu \in \calT{}$. Set $\Lambda_\mu$ contains a curve $\lambda_\mu(e)$ connecting $x_\mu(e)$ with $y_\mu(e)$, for each inter-cluster edge $e$ incident to a \me component of $\mu$.
	The curves in $\Lambda_\mu$ are drawn as simple curves in the interior of cycle $\blue{C^2_\mu}$ so that (i) they do not cross each other, (ii) they do not cross any of the curves representing edges $\red{(v,x_\mu(e'))}$ and edges $\red{(y_\mu(e'),y_\nu(e'))}$, for every vertex $v \in \mu$ and for every inter-cluster edge $e'$ incident to $\mu$, and (iii) they do not cross any of the edges of $\Gint{}$ between two vertices of the copy of a component $c_i$ belonging to $\mu$.
	This is always possible, since $\Theta(X_\mu)=\overline{\Phi(Y_\mu)}$, where we set $x_\mu(e)=y_\mu(e)$. We give a proof of this claim. First, the matching in $\Eb{}$ between the leaves of $X_\mu$ and those of $B_\mu$ ensures that $\Theta(X_\mu)=\overline{\Theta(B_\mu)}$. Analogously, the matching in \Er{} between the leaves of $B_\mu$ and those of $A_\mu$ ensures that $\overline{\Theta(B_\mu)}=\Theta(A_\mu)$. By repeating this argument while considering matchings in either \Eb{} or \Er{} between the leaves of $A_\mu$ and of $A'_\mu$, the leaves of $A'_\mu$ and of $B'_\mu$, and the leaves of $B'_\mu$ and the leaves of $Y_\mu$ corresponding to inter-cluster edges incident to a \me component of $\mu$, we have that 
	$\Theta(X_\mu)=\overline{\Theta(B_\mu)}=\Theta(A_\mu)=\overline{\Theta(A'_\mu)}=\Theta(B'_\mu)=\overline{\Phi(Y_\mu)}$.
	
	We now draw each inter-cluster edge $e=(u,v)$ in $\Gamma$, where $u \in \mu$ and $v \in \nu$.
	If $e$ is incident to a \me component of $\mu$ and to a \me component of $\nu$, it is drawn as a composition of five parts. The first and the last parts of $e$ coincide with the drawing of edge $\red{(u,x_\mu(e))} \in \Er{}$ of $G_\mu$ and of edge $\red{(v,x_\nu(e))} \in \Er{}$ of $G_\nu$ in $\GammaR{}$, respectively. The second and the fourth part coincide with curves $\lambda_\mu(e) \in \Lambda_\mu$ and $\lambda_\nu(e) \in \Lambda_\nu$, respectively. Finally, the middle part coincides with the drawing of edge $\red{(y_\mu(e),y_\nu(e))} \in \Er{}$ in $\GammaR{}$. 
	If $e$ is incident to a \se component of $\mu$ (of $\nu$), then the first and the second part (the fourth and the fifth part) are not drawn.
	
	Finally, for each \se component $c_i$ of $\mu$, let $e=(v,u)$ be the unique inter-cluster edge incident to $c_i$, with $v \in c_i$. We add to $\Gamma$ a planar drawing of $c_i$ in which $v$ is incident to the outer face, so that $v$ lies in the same position as $y_\mu(e)$ in \sefesolution{} and there exists no crossing involving an edge of $c_i$.
	
	We now prove that $\Gamma$ is a c-planar drawing. Recall that we constructed region $R(\mu)$ for each cluster $\mu$ so that its boundary $B(\mu)$ coincides with $\blue{C^2_\mu}$ in \GammaB{}. This implies that $R(\mu)$ contains all and only the vertices of $\mu$, since all the vertices of the copies of the components $c_i$ of $\mu$, which belong to $G_\mu$, lie inside $\blue{C^2_\mu}$ and since all the vertices of \sefeinstance{} not in $G_\mu$ lie in the exterior of any cycle of $\Gb{}$.
	
	Also, there exists no region-region crossings in $\Gamma$, since $\GammaB{}$ is a planar drawing of $\Gb{}$, and since $\blue{C^2_\mu}$ and $\blue{C^2_\nu}$ are vertex disjoint cycles in $\Gb{}$, for each $\mu, \nu \in \calT{}$.
	
	Further, there exists no edge-region crossing in $\Gamma$. In fact, the only intersection between $B(\mu)$, for each cluster $\mu \in \calT{}$, and an edge of $G$ is on the portion of $B(\mu)$ corresponding to edge $\blue{(\rho^2_\mu,\rho^7_\mu)}$, since the remaining portion of $B(\mu)$ corresponds to edges in $\Gint{}$, which are not crossed in \sefesolution{}.
	Also, edge $\blue{(\rho^2_\mu,\rho^7_\mu)}$ is only crossed (once) by edges in $\Gr{}$ between a leaf of $Y_\mu$ and a leaf of $Y_\nu$, for some $\nu \in \calT{}$. Hence, for each inter-cluster edge $e$ incident to $\mu$, only one of the five curves that have been used to draw $e$ crosses $B(\mu)$, namely the middle one, and hence every edge of $G$ crosses $B(\mu)$ at most once.
	
	Finally, there exists no edge-edge crossing in $\Gamma$.
	Namely, observe that each edge $e$ in $G$ is either represented by an edge in $\Gr{}$ (if $e$ is an intra-cluster edge) or by the composition of edges in $\Gr{}$ and curves in $\Lambda_\mu$ and $\Lambda_\nu$, where $\mu$ and $\nu$ are the clusters $e$ is incident to (if $e$ is an inter-cluster edge).
	Hence, the planarity of the drawing of $G$ in $\Gamma$ descends from the  planarity of $\GammaR{}$ and from the construction of the sets $\Lambda_\mu$ and $\Lambda_\nu$.
	
	We finally prove that $\Gamma$ is consistent with $\langle \calX{}, \calY{} \rangle$.
	Since the only edge of $\blue{C^2_\mu}$ that is crossed in \sefesolution{} is $\blue{(\rho^2_\mu,\rho^7_\mu)}$, the linear order $\calO{}_\mu$ in which the edges incident to $\mu$ cross $B(\mu)$ in $\Gamma$, starting from $e_\mu$, coincides with the linear order in which the edges in $\Gr{}$ cross $\blue{(\rho^2_\mu,\rho^7_\mu)}$ in \sefesolution{}, starting from $e_\mu$.
	By the planarity of $\GammaR{}$, this order coincides with the reverse of $\Theta(Y_\mu)$, when we set $e = y_\mu(e)$.
	Hence, for every internal node $\nu$ of $Y_\mu$, all the leaves of the subtree of $Y_\mu$ rooted at $\nu$ appear consecutively in $\calO{}_\mu$, and thus $\Gamma$ is consistent with $Y_\mu$.
	Analogously, $\Gamma$ is consistent with $X_\mu$, since $\Phi(Y_\mu)=\overline{\Theta(X_\mu)}$.
	Repeating this argument for all the clusters $\mu \in \calT{}$ proves the statement.
	
	$(\Longleftarrow)$ Suppose that $\cgraph{}$ admits a c-planar drawing $\Gamma$ that is consistent with $\langle \calX{}, \calY{} \rangle$.
	We show how to construct a SEFE \sefesolution{} of \sefeinstance{}. By Theorem~\ref{th:inutile}, we can describe \sefesolution{} by means of the embeddings $\EmbR{}$ and $\EmbB{}$ of $\Gr{}$ and of $\Gb{}$, respectively.
	
	In the following we assume that edge $e_{\mu_1}$ (and hence vertex $\rho^1_{\mu_1}$) is incident to the outer face of the drawing of $G$ in $\Gamma$. This is possible since $e_{\mu_1}$ is an inter-cluster edge.
	
	We construct $\EmbR{}$ and $\EmbB{}$ in such a way that cycle $C$ bounds a face, which we assume to be the outer face in both $\EmbR{}$ and $\EmbB{}$. 
	Clearly, this uniquely determines the rotation scheme of $\sigma_\mu$ and $\rho^1_\mu$ in $\EmbB{}$, for each cluster $\mu \in \calT{}$, and of $\sigma^*$ and $\rho^1_{\mu_1}$ in $\EmbR{}$.
	Further, this implies that wheel $W_\mu$, for each $\mu \in \calT{}$, must be embedded so that $c_\mu$ lies in the interior of its rim in both $\EmbR{}$ and $\EmbB{}$.
	We will embed all the other vertices in $V$ and edges in $\Er{}$ and $\Eb{}$ so that they lie in the exterior of the rim of each $W_\mu$ in both $\EmbR{}$ and $\EmbB{}$. This uniquely determines the rotation scheme of all vertices $\rho^2_\mu, \dots, \rho^7_\mu$ in $\EmbB{}$ and $\EmbR{}$.
	
	Let $\cal{O}_\mu$ be the clockwise linear order in which the inter-cluster edges incident to $\mu$ cross $B(\mu)$ in $\Gamma$, starting from $e_\mu$. We set the rotation scheme of the other vertices of $G_\mu$ so that:
	\begin{inparaenum}
		\item $\overline{\Theta(Y_\mu)}$ coincides with $\calO{}_\mu$ in both $\EmbR{}$ and in $\EmbB{}$,
		\item the clockwise order of the paths connecting the center of star $B'_\mu$ with the leaves of tree $Y_\mu$ in $\Gb{}$ not passing through $\beta'_\mu$ coincides with $\overline{\Theta(Y_\mu)}$ in $\EmbB{}$, when we identify each path with the leaf of $Y_\mu$ it is incident to,
		\item $\Theta(B'_\mu)$ coincides with $\overline{\Phi(Y_\mu)}$ in $\EmbB{}$,
		\item each of $\overline{\Theta(A'_\mu)}$,$\Theta(A_\mu)$,$\overline{\Theta(B_\mu)}$, and $\Theta(X_\mu)$ coincides with $\overline{\Phi(Y_\mu)}$ in both $\EmbR{}$ and in $\EmbB{}$,
		\item each vertex $v$ of the copy of a \me component $c_i$ of $\mu$ has the same rotation scheme in $\EmbR{}$ as the corresponding vertex in $\Gamma$, where we replace $e$ with $(v,x_\mu(e))$, if $e$ is an inter-cluster edge incident to $\mu$;
		\item each vertex $v$ of the copy of a \me component $c_i$ of $\mu$ has the same rotation scheme in $\EmbB{}$ as the corresponding vertex in $\Gamma$, where we remove all of the inter-cluster edges incident to $v$, except for one edge $e_v$, which we replace with $(v,z(c_i))$;
		\item for each vertex $z(c_i)$ of $C_\mu$, the order of the edges in the rotation scheme of $z(c_i)$ in $\EmbB{}$ is the same as the order in which these edges appear in a counter-clockwise walk around the boundary of $c_i$ in $\Gamma$, where we remove all of the inter-cluster edges incident to $v$, except for edge $e_v$;
		\item the center of $C_\mu$ has any rotation scheme in both $\EmbR{}$ and in $\EmbB{}$;
		\item edges $\red{(v,x_\mu(e_\mu))}$, $(\delta'_\mu,x_\mu(e_\mu))$, and $(\rho_\mu,x_\mu(e_\mu))$ appear in this order in the rotation scheme of $x_\mu(e_\mu)$ in $\EmbR{}$, where $v$ is the vertex of $\mu$ edge $e_\mu$ is incident to;
		\item edges $\red{(b_\mu(e_\mu),x_\mu(e_\mu))}$, $(\delta'_\mu,x_\mu(e_\mu))$, and $(\rho_\mu,x_\mu(e_\mu))$ appear in this order in the rotation scheme of  $x_\mu(e_\mu)$ in $\EmbB{}$;
		\item edges $(\xi_\mu,y_\mu(e_\mu))$, $(\delta_\mu,y_\mu(e_\mu))$, and $\red{(y_\nu(e_\mu),y_\mu(e_\mu))}$ appear in this order in the rotation scheme of $y_\mu(e_\mu)$ in $\EmbR{}$, where $\nu$ is the other cluster to which $e_\mu$ is incident;
		\item edges $(\xi_\mu,y_\mu(e_\mu))$, $(\delta_\mu,y_\mu(e_\mu))$, and $\red{(b'_\mu(e_\mu),y_\mu(e_\mu))}$ appear in this order in the rotation scheme of $y_\mu(e_\mu)$ in $\EmbB{}$.
	\end{inparaenum}
	Note that the rotation scheme of the remaining vertices of $G_\mu$ in $\EmbR{}$ and $\EmbB{}$ (namely the leaves of stars $C_\mu$, $B_\mu$, $A_\mu$, $B'_\mu$, and $A'_\mu$, and the leaves of trees $X_\mu$ and $Y_\mu$ different from $x_\mu(e_\mu)$ and from $y_\mu(e_\mu)$, respectively) is unique, since they have degree less or equal $2$ in $\Gr{}$ and $\Gb{}$.
	
	We prove that both $\EmbR{}$ and $\EmbB{}$ are planar. 
	First note that there exists a planar embedding of $X_\mu$ and of $Y_\mu$ so that $\overline{\Theta(Y_\mu)}=\calO{}_\mu$ and $\Theta(X_\mu)=\overline{\Phi(Y_\mu)}$, since $\Gamma$ is consistent with $X_\mu$ and with $Y_\mu$. 
	The embedding of the biconnected components of $\Gb{}$ induced by the vertices
	(i) of $X_\mu$ and of $B_\mu$,
	(ii) of $A_\mu$ and of $A'_\mu$, and
	(iii) of $B'_\mu$ and of $Y_\mu$
	are planar 
	since  $\Theta(X_\mu)=\overline{\Theta(B_\mu)}$,
	since  $\Theta(A_\mu)=\overline{\Theta(A'_\mu)}$, 
	and since 
	the clockwise order of the paths connecting the center of star $B'_\mu$ with the leaves of tree $Y_\mu$ in $\Gb{}$ not passing through $\beta'_\mu$ coincides with $\overline{\Theta(Y_\mu)}$ in $\EmbB{}$.
	
	Analogously, the embedding of the biconnected components of $\Gr{}$ 
	induced by the vertices
	(i) of $A_\mu$ and of $B_\mu$, and
	(ii) of $A'_\mu$ and of $B'_\mu$,
	are planar 
	since $\Theta(A_\mu)= \overline{\Theta(B_\mu)}$, 
	and since $\Theta(B'_\mu)=\overline{\Theta(A'_\mu)}$, respectively.
	
	Also, the embedding of the biconnected component of $\Gb{}$ composed of the copy of each \me component $c_i$ of $\mu$, of vertex $z(c_i)$, and of the edges between them is planar, by the construction of the rotation scheme of $z(c_i)$.
	Further, the embedding of the subgraph of $\Gr{}$ composed of the copies of all the \me components $c_i$ of $\mu$, of tree $X_\mu$, and of the edges between them is planar since $\Theta(X_\mu)$ coincides with $\calO{}_\mu$ restricted to the inter-cluster edges incident to the \me components of $\mu$.
	Finally, the embedding of the biconnected component of $\Gr{}$ composed of tree $Y_\mu$, of tree $Y_\nu$, and of the edges between their leaves, for each two adjacent clusters $\mu$ and $\nu$ in $G_A$, is planar since $\overline{\Theta(Y_\mu)}=\Theta(Y_\nu)$, restricted to the inter-cluster edges incident to both $\mu$ and $\nu$.
	This is due to the fact that $\calO{}_\mu=\overline{\Theta(Y_\mu)}$, that $\calO{}_\nu=\overline{\Theta(Y_\nu)}$, and that $\calO{}_\mu$ coincides with the reverse of $\calO{}_\nu$, when both orders are restricted to the edges incident to both $\mu$ and $\nu$, by the c-planarity of $\Gamma$.
	Note that, since $\Gamma$ has edge $e_{\mu_1}$ on the outer face, vertex $\rho^1_{\mu_1}$ is not enclosed by any cycle of $\Gr{}$, except for $C$. Hence, vertices $\rho^1_{\mu_1}$ and $\sigma^*$ are incident to the same face of $\EmbR{}$.
	
	The planarity of $\EmbR{}$ and of $\EmbB{}$, restricted to the edges of $G_\mu$ in $\Er{}$ and in $\Eb{}$, respectively, is implied by the planarity of the embedding of each of the above considered components of $\Gr{}$ and $\Gb{}$, and by the order in which $\delta_\mu$, $\alpha_\mu$, $\beta_\mu$, $\gamma_\mu$, $\delta'_\mu$, $\alpha'_\mu$, and $\beta'_\mu$ appear along the rim of $W_\mu$.
	
	Further, since each $G_\mu$ is only connected to the frame cycle $C$ via edge $\blue(\rho^1_\mu,\sigma_\mu)$, the planarity of $\EmbB{}$ restricted to the edges of each $G_\mu$ in $\Eb{}$ implies the planarity of the whole $\EmbB{}$.
	To complete the proof of the planarity of $\EmbR{}$, it only remains to consider the embedding of the subgraph of $\Gr{}$ induced by the vertices of all trees $Y_\mu$, with $\mu \in \calT{}$. The planarity of this subgraph descends from the planarity of the embedding of the subgraph of $\Gr{}$ induced by the vertices of each two trees $Y_\mu$ and $Y_\nu$ such that $\mu$ and $\nu$ are adjacent in $G_A$, from the fact that $\Gamma$ is consistent with $Y_\mu$, for each $\mu \in \calT{}$, and from the c-planarity of $\Gamma$.
	This completes the proof that \sefesolution{} is a SEFE of \sefeinstance{}.
	
	We conclude the lemma by proving that \sefeinstance{} can be transformed into an equivalent instance in which $G_\cap=(V,\Er{} \cap \Eb{})$, $\Gr{}$, and $\Gb{}$ satisfy the required properties.
	
	We note that $\Gr{}$ is connected, since $G_A$ is connected. Also, $\Gr{}$ contains cut-vertices $\sigma^*$ and $\rho^1_{\mu_1}$. Further, for each cluster $\mu \in \calT{}$, $\Gr{}$ contains cut-vertices $\delta_\mu$, $y_\mu(e_\mu)$, $\delta'_\mu$, $x_\mu(e_\mu)$, $\gamma_\mu$, the center of star $C_\mu$, the internal vertices of $X_\mu$, and possibly the internal vertices of $Y_\mu$. We now show that all these cut-vertices are incident to at most two non-trivial blocks of $\Gr{}$.
	Namely, vertices $\sigma^*$, $\rho^1_{\mu_1}$, and vertices $\delta_\mu$, $\delta'_\mu$, $\gamma_\mu$, $x_\mu(e_\mu)$, and $y_\mu(e_\mu)$, for each cluster $\mu \in \calT{}$,  are incident to exactly two blocks in $\Gr{}$.
	The center of star $C_\mu$ is incident only to non-trivial blocks.
	Each internal vertex $c_i$ of $X_\mu$ is incident to at most one non-trivial block, that is, the block composed of vertex $c_i$, of the leaves of $X_\mu$ incident to $c_i$, and of the vertices of the copy of the \me component $c_i$ in $G_\mu$.
	Each internal vertex $\nu$ of $Y_\mu$ is incident to at most one non-trivial block, that is, the one composed of $\nu$, of the leaves of $Y_\mu$ incident to $\nu$, of the vertex $\mu$ in $Y_\nu$, and of the leaves of $Y_\nu$ incident to $\mu$.
	
	We note that $\Gb{}$ is connected, by construction. Also, for each cluster $\mu \in \calT{}$, graph $\Gb{}$ contains cut-vertices $\gamma_\mu$, the center of star $C_\mu$, and vertices $z(c_i)$, for each component $c_i$ of $\mu$.
	We now show that all these cut-vertices are incident to at most two non-trivial blocks of $\Gb{}$. Namely, vertex $\gamma_\mu$ and vertices $z(c_i)$, for each \me component $c_i$ of $\mu$, are incident to exactly two blocks, while the center of $C_\mu$ is incident only to non-trivial blocks.
	
	Finally, no vertex of a copy of a \me component $c_i$ of $\mu$ is a cut-vertex in either $\Gr{}$ or $\Gb{}$. This is due to the fact that, by assumption, every block of $c_i$ that is a leaf in the BC-tree of $c_i$ has at least an inter-cluster edge incident to one of its vertices that is not a cut-vertex of $c_i$. Hence, all the vertices of the copy of $c_i$, together with the vertex $c_i \in X_\mu$ and with the leaves of $X_\mu$ incident to it, belong to the same block of $\Gr{}$. Also, all the vertices of the copy of $c_i$, together with vertex $z(c_i)$, belong to the same block of $\Gb{}$.
	
	By Claim~\ref{cl:cycle-removal}, the cycles of $\Gint{}$ can be removed so that $\Gint{}$ becomes a forest, without altering the properties of \sefeinstance{}. Observe that the only cycles contained in $\Gint{}$ are the frame cycle $C$, the cycles in $W_\mu$, for each $\mu \in \calT{}$, and (possibly) the cycles in the copy of some \me component $c_i$ of $G_\nu$, for some cluster $\nu \in \calT{}$. 
	This concludes the proof of the lemma.
\end{proof}

We conclude by exploiting Lemma~\ref{le:inc-constr-cp-2-SEFE} to prove the main result of the section.

\rephrase{Theorem}{\ref{th:algo-inc-constr-cp}}
{An instance $\langle \calC{}(G(V,E),\calT{}), \calX{}, \calY{} \rangle$ of \iccp can be tested in $O(|V|^2)$ time.
}
\begin{proof}
	Since for each inter-cluster edge $e$ there exist at most two \ctrees in $\calX{}$ and exactly two \nctrees in $\calY{}$ with a leaf corresponding to $e$, we have $|\calX{}|,|\calY{}| \in O(|E|)$.
	Also, since $G$ is planar, $|E| \in O(|V|)$, and $|G| \in O(|V|)$. Hence, $s = |G|+|\calX{}| + |\calY{}| \in O(|V|)$.
	
	Apply Lemma~\ref{le:inc-constr-cp-2-SEFE} to $\langle \calC{}(G(V,E),\calT{}), \calX{}, \calY{} \rangle$ to construct in $O(s)$ time an equivalent instance of SEFE 
	satisfying the conditions of Theorem~\ref{th:sefe}, which can be tested in $O(s^2) \in O(|V|^2)$ time.  
\end{proof}


\section{Conclusions and Open Problems}

In this paper we studied the problem of constructing c-planar drawings with pipes of flat c-graphs. We presented algorithms to test the existence of such drawings when the number of certain components is small, in different scenarios, namely when the \ga is a path ({\sc Strip Planarity}), when it has a fixed embedding ({\sc C-Planarity with Embedded Pipes}), and when it has no restrictions ({\sc C-Planarity with Pipes}).

Several questions are left open. We find particularly interesting to determine whether there exist combinatorial properties of the nesting of the components that would allow us to reduce the number of possible \ctrees, analogous to the ones we could prove for the \enctrees. We remark that the introduction of the \ctrees already allowed us to make the running time of our algorithms, and in particular of the FPT algorithm, independent of the size of each component.

Another natural question concerns the possibility of extending our results to problem \cp. An important goal would be to determine the complexity of this problem for flat c-graphs in the case in which each cluster contains at most two components. Efficient algorithms for this case exist only when the underlying graph has a fixed embedding~\cite{jjkl-cpecgtcc-08}, when also each co-cluster has at most two components~\cite{br-npcpcep-14}, or when the cut-vertices of the \ga have at most two non-trivial blocks~\cite{br-npcpcep-14}. 

We would like to point out that this latter result is obtained by considering a graph that is in fact the one we defined as \ga. Namely, the authors of~\cite{br-npcpcep-14} introduced a data structure, called CD-tree, which is a star when the considered c-graph \cgraph{} is flat; in this case, the skeleton associated to the central node of this star turns out to coincide with the \ga $G_A$ of \cgraph{}.
In this paper~\cite{br-npcpcep-14}, problem \cp for flat c-graphs is described in terms of a specific constrained-planarity problem for $G_A$, namely the problem of computing a planar embedding of this graph satisfying a set of partitioned PQ-constraints. The mentioned result for flat c-graphs is then obtained by showing that the given restrictions for the original c-graph allow to generate instances of this constrained-planarity problem that can be solved by means of the Simultaneous PQ-ordering framework~\cite{br-spacep-13}. The authors also extended their result to give an FPT algorithm for the same problem in two parameters that depend on the total number of clusters and on the number of edges leaving a cluster. We remark that analogous results (with slightly different parameters for the FPT algorithm) could be obtained using the techniques of our paper; a key property for this is the fact that, when $G_A$ is biconnected, the \nctree of each cluster can be proved to be unique. We thus ask whether deeper considerations on the possible nesting configurations of the clusters could be used to further reduce the number of \nctrees to be considered even when the cut-vertices of $G_A$ have a larger number of non-trivial blocks.

\bibliographystyle{plainurl}
\bibliography{paper132}

\end{document}